\documentclass[letterpaper, 11pt]{article}

\usepackage{amsmath, amssymb, amsfonts, amsthm}
\usepackage{bm}
\usepackage{mathtools}
\usepackage{mathrsfs}
\usepackage{geometry}
\usepackage{setspace}
\usepackage[utf8x]{inputenc}
\usepackage{algpseudocode,algorithmicx,algorithm}
\usepackage[T1]{fontenc}
\usepackage{lmodern}
\usepackage{graphicx}
\usepackage{enumerate}
\usepackage[natural, dvipsnames, pdftex]{xcolor}
\usepackage{tikz}
\usepackage{verbatim}
\usepackage{hyperref}
\usepackage{cancel}
\usepackage[all]{xy}
\usepackage{mdframed}
\usepackage[capitalize]{cleveref}
\usepackage{todonotes}
\usepackage{algpseudocode}

\theoremstyle{plain}
\newtheorem{thm}{Theorem}[section]
\newtheorem{lemma}[thm]{Lemma}

\newtheorem{claim}[thm]{Claim}

\theoremstyle{definition}
\newtheorem{defn}[thm]{Definition}

\newtheorem{col}[thm]{Corollary}

\theoremstyle{remark}
\newtheorem{rmk}[thm]{Remark}

\newcommand{\Z}{\mathbb{Z}}

\newcommand{\F}{\mathbb{F}}

\newcommand{\bv}{\mathbf{v}}
\newcommand{\bx}{\mathbf{x}}
\newcommand{\by}{\mathbf{y}}

\newcommand{\bc}{\mathbf{c}}

\newcommand{\cC}{\mathcal{C}}

\begin{document}

\title{List Decoding of Reed--Solomon Codes and Folded Reed--Solomon Codes Over Galois Ring
\thanks{
C. Yuan is with School of Computer Science, Shanghai Jiao Tong University. (Email: \href{chen_yuan@sjtu.edu.cn}{chen\_yuan@sjtu.edu.cn}) R. Zhu is with School of Computer Science, Shanghai Jiao Tong University. (Email: \href{sjtuzrq7777@sjtu.edu.cn}{sjtuzrq7777@sjtu.edu.cn})
}}

\author{Chen Yuan, Ruiqi Zhu}
\date{}
\maketitle

\begin{abstract}
List decoding of codes can be seen as the generalization of unique decoding of codes. While list decoding over finite fields has been extensively studied, extending these results to more general algebraic structures such as Galois rings remains an important challenge. Due to recent progress in zero knowledge systems, there is a growing demand to investigate the proximity gap of codes over Galois rings~\cite{JiaLiXingYaoYuan2025PCS_GaloisRings, GLSTW23, wei2025transparent}. The proximity gap is closely related to the decoding capability of codes. It was shown~\cite{BCIKS20} that the proximity gap for RS codes over finite field can be improved to $1-\sqrt{r}$ if one consider list decoding instead of unique decoding. However, we know very little about RS codes over Galois ring which might hinder the development of zero knowledge proof system for ring-based arithmetic circuit. In this work, we first extend the list decoding procedure of Guruswami and Sudan to Reed-Solomon codes over Galois rings, which shows that RS codes with rate $r$ can be list decoded up to radius $1-\sqrt{r}$. Then, we investigate the list decoding of folded Reed-Solomon codes over Galois rings. We show that the list decoding radius of folded Reed-Solomon codes can reach the Singlton bound as its counterpart over finite field. We also extend the deterministic pruning method of~\cite{ashvinkumar2026algorithmic} to Galois rings, showing how to prune the affine free module obtained from the linear-algebraic decoder and recover the candidate codewords. Finally, we obtain an \emph{algorithmic} list-size bound of $O(1/\varepsilon^2)$ for our folded Reed--Solomon code by extending the approach of~\cite{srivastava2025improved} to Galois rings. Moreover, at the \emph{combinatorial} level, by developing the recent work of~\cite{chen2025explicit}, we show that folded Reed--Solomon codes over Galois rings satisfy the relaxed generalized Singleton bound in the average-radius sense with optimal list size $O(1/\varepsilon)$. Specifically, this result is not obtained by a straightforward extension of~\cite{chen2025explicit} from finite fields to Galois rings; instead, we develop a more delicate tree-based argument that exploits the $p$-adic congruence structure of Galois rings.
\end{abstract}

\section{Introduction}
List decoding, first introduced in~\cite{elias1957list}, provides a way to recover codewords even when the number of errors $e$ goes beyond half of the minimum distance $d$. Specifically, if the number of errors $e$ in a received word exceeds $\lfloor (d-1)/2\rfloor$, it is possible that more than one codeword that is within \emph{(Hamming)} distance $e$ from the received word. In this case, a list decoder outputs all codewords that fall within this \emph{Hamming} ball of radius $e$. 

Reed-Solomon codes (RS codes for short), were first proposed in 1960~\cite{reed1960polynomial}. RS codes belong to a family of Since RS codes belong to the family of the maximum distance separable (MDS) codes. RS codes also have very efficient encoding and decoding algorithms~\cite{berlekamp2015algebraic} and~\cite{sugiyama1975method}. Let $\rho$ be the decoding radius and $R$ be the rate of a code. Sudan~\cite{sudan1997decoding} introduced the first explicit list decoding algorithm for RS codes that can decoded RS codes beyond unique decoding radius. Subsequently, Guruswami and Sudan~\cite{guruswami1998improved} refined that algorithm to achieve Johnson bound for any rate. Furthermore, their method can also be extended to the decoding of algebraic geometry codes which initiated an intensive line of research that produced numerous results in the field of list decoding~\cite{koetter2003algebraic,pellikaan2004list,roth2002efficient,tal2003list}. Understanding the limits of list-decoding and list-recovery of RS codes is of prime interest in coding theory and has attracted a lot of attention over the past decades. In a recent breakthrough, Shangguan and Tamo proved that~\cite{ST20}, the random RS codes can approach the generalized Singlton bound for list size $L=2,3$ which is far beyond the Johnson bound. Brakensiek, Gopi and Makam~\cite{brakensiek2023generic} further showed that such results hold for any list size. We note that If we relax the generalized Singlton bound with $\epsilon$ gap, then the field size can be optimized to $O(\frac{n}{\epsilon})$~\cite{guo2023randomly,alrabiah2024randomly}. We note that all these results about RS codes beyond Johnson bound is combinatorial which means there is no explicit algorithm to construct such codes and also lacks of no efficient encoding and decoding algorithm.   

To explicitly decodes code up to Singleton bound, we need to deviate from RS codes. Building upon the prior work of~\cite{parvaresh2005correcting}, Guruswami and Rudra~\cite{guruswami2008explicit} presented the first explicit construction of codes called folded Reed-Solomon codes (FRS codes) with list decoding radius approaching Singleton bound. Subsequent works improved the decoding algorithms and the corresponding list-size bounds~\cite{guruswami2011linear,vadhan2012pseudorandomness,guruswami2011linear,dvir2012subspace}. Kopparty, Ron-Zewi, Saraf and Wooters~\cite{kopparty2023improved} managed to bring down the list size of FRS codes to constant $(\frac{1}{\varepsilon})^{O(\frac{1}{\varepsilon})}$. Srivastava~\cite{srivastava2025improved} showed explicit folded RS codes with rate $R$ that can be list decoded up to radius $1-R-\varepsilon$ with lists of size $O(\frac{1}{\varepsilon^2})$. More recently, Chen and Zhang~\cite{chen2025explicit} proved a relaxed generalized Singleton bound with optimal list size $O(\frac{1}{\varepsilon})$ which fully resolvs a long-standing open problem proposed by Guruswami and Rudra. A recent work of Ashvinkumar, Habib and Srivastava~\cite{ashvinkumar2026algorithmic} developed deterministic pruning algorithms for folded Reed--Solomon codes over finite fields. Their pruning procedure starts from the affine subspace produced by the linear-algebraic decoder and filters it to the exact list.

Despite of so many progress made in the list decoding of codes over finite fields, there are very few works considering the list decoding over rings. One reason is due to that codes over finite field is considered to be superior to codes over rings. Moreover, there is very few applications for codes over rings. However, there is a trend to design good codes over rings due to recent progress of zero knowledge proof.
An efficient zero knowledge proof system such as SNARKs requires a codes with large decoding radius whether unique decoding or list decoding.  
There are zero knowledge proof systems~\cite{huang2025sublinear,ChiesaFioreMicali2023Rinocchio,JiaLiXingYaoYuan2025PCS_GaloisRings} defined over rings which can handle the arithmetic circuit over $\Z_{2^k}$ without expensive translations of statement from finite field. A code with large decoding radius indicates a large proximity gap which is crucial to the analysis of soundness error for the zero knowledge proof system. It was shown that the proximity gap can be improved from $\frac{1-r}{2}$ to $1-\sqrt{r}$ if we consider the list decoding instead of unique decoding for RS codes over finite fields~\cite{BCIKS20}. However, when migrated to RS codes over Galois ring, the state-of-the-art result is a proximity gap $\frac{1-r}{2}$~\cite{JiaLiXingYaoYuan2025PCS_GaloisRings}. Thus, to improve the performance of zero knowledge system over rings, it is of great interest to investigate the codes over rings.

In this paper, we generalize most of the state-of-the-art techniques about the list-decodable codes to Galois rings including the celebrated Guruswami-Sudan list decoding algorithm, the list decoding algorithm of Folded Reed-Solomon codes and its improved list size analysis. In addition, we extend the deterministic pruning approach of~\cite{ashvinkumar2026algorithmic} to Galois rings by lifting candidates through the quotient rings $GR(p^t,\ell)$ and reducing each lifting step to a finite-field pruning problem over the residue field.

\subsection{Related works}
\paragraph{Decoding over rings.}While most of these advances have been developed over finite fields, recent research has highlighted the importance of extending coding theory to more general algebraic structures such as rings. Specifically, Galois rings provide a rich algebraic framework that has found applications in networking and zero-knowledge proofs~\cite{renner2021low,wei2025transparent,lin2024more}.

Unique decoding of codes over Galois rings has been studied in several works. For example, Babu and Zimmermann~\cite{babu2001decoding} gave decoding methods for linear codes over Galois rings. In the list-decoding setting, Armand~\cite{armand2005list} showed that the Guruswami--Sudan procedure can be used to decode generalized Reed--Solomon codes over commutative rings with identity, and later improved the list decoding of generalized Reed--Solomon and alternant codes over Galois rings~\cite{armand2005improved}. The latter work gives a two-stage decoder based on Guruswami--Sudan decoding and studies the probability of successful decoding beyond the Guruswami--Sudan radius, but does not analyze the resulting list size.

Quintin~\cite{quintin2012algorithms} further studied algorithmic aspects of Guruswami--Sudan list decoding over finite rings, including Galois rings, with an emphasis on complexity and lifting-based decoding methods. Our work differs from these prior works in two aspects. First, we extend the Guruswami--Sudan list-decoding algorithm for Reed--Solomon codes over Galois rings and obtain the Johnson-bound list-size guarantee. Second, we study folded Reed--Solomon codes over Galois rings, analyze the free-module structure of the candidate space, and prove deterministic pruning and list-size bounds for decoding up to radius $1-R-\varepsilon$, approaching the list-decoding capacity.

\paragraph{Applications to Zero Knowledge Proof.} Exploring codes over Galois ring have direct implications for zero-knowledge proof (ZKP) systems, particularly succinct non-interactive arguments of knowledge (SNARKs) and scalable transparent arguments of knowledge (STARKs). ZKP systems are cryptographic protocols that enable a prover to convince a verifier of the validity of a statement without revealing any information beyond its truth. In recent years, the design of efficient ZKPs—particularly SNARKs and STARKs—has become deeply connected to coding theory~\cite{ben2016interactive,ben2017scalable,arnon2024stir,zhang2024fast,chiesa2020fractal,rothblum2013interactive}. In these systems, Reed–Solomon (RS) codes and its variants serve as the mathematical foundation for low-degree testing and proximity proofs, which are essential for ensuring soundness and succinctness. By treating polynomial evaluations as RS codewords, the problem of verifying the low-degree property of a function can be reduced to test its proximity to a codeword. This algebraic connection underpins many modern proof systems, including FRI-based STARKs and code-based polynomial commitment schemes~\cite{ben2018fast}.

Despite the rapid development of zero knownedge proof, there still remains a gap between theoretical studies and practical usage. For example, most SNARKs focus on the field arithmetic, which means that statements are modeled as arithmetic circuits over a finite field. While for real-life applications, there is a growing demands for statements represented by ring arithmetic. A direct solution is to emulate the binary operations as the field operations. However, this would introduce a significant overhead. Thus, it is necessary to design the zero knowlege proof system over rings. The Rinocchio protocol~\cite{GNS23} was the first complete SNARKs protocols designed for ring-based arithmetic circuits. et al.,~\cite{JiaLiXingYaoYuan2025PCS_GaloisRings} noted that the Rinocchio protocol follows the paradigm of linear probabilistically checkable proofs (linear PCPs) which has some downside such as large prover computation, designated-verifier and trusted setups. Thus, they proposed a polynomial commitment scheme based on RS codes over Galois ring. Combined with polynomial interactive oracle proofs, they obtained a publicly verifiable SNARKs over $\mathbf{Z}_{2^k}$. Recently, Wei, Zhang and Deng~\cite{wei2025transparent} proposed transparent SNARKS over Galois ring which extend Brakedown~\cite{GLSTW23} commitment scheme to Galois rings. We note that the RS codes over Galois ring is the key ingredient of polynomial commitment scheme in~\cite{JiaLiXingYaoYuan2025PCS_GaloisRings,GLSTW23}. Thus, it is worth exploring the performance of codes over Galois ring.

We note that extending the optimal finite-field list-size bound of~\cite{chen2025explicit} to Galois rings is more delicate than extending the linear-algebraic framework of~\cite{srivastava2025improved}. A naive attempt to construct a ring-valued geometric hypergraph leads to technical gaps. Our approach avoids this issue by organizing a bad list through a $p$-adic congruence tree, passing each effective internal node to a residual folded Reed--Solomon instance over the residue field, and then applying the finite-field contradiction lemma locally before summing the resulting inequalities globally.

\subsection{Our Contributions}

As mentioned above, we obtained list decoding algorithms for RS codes and FRS codes over Galois ring. We start with the list decoding algorithm for RS codes.

\paragraph{List Decoding Algorithm for RS codes.}

For RS codes, we first exploit the property of unique quasi-prime factorization in Galois rings to characterize the explicit form of the linear factors of polynomials. For a polynomial $f(x)\in GR(p^a,\ell)[x]$, if $p\nmid f(x)$ and its reduction $\bmod p$ can be factored in the residue field $\F_{p^\ell}[x]$ as $\overline{f(x)}=(x-a_1)^{\ell_1}(x-a_2)^{\ell_2}\ldots (x-a_s)^{\ell_s}\overline{g(x)}$, where $g(x)$ has no linear factor over $\F_{p^\ell}$. Then we lift each factor $(x-a_i)$ and $\overline{g(x)}$ to $GR(p^a,\ell)[x]$ as:
\[
f(x)=(x-c_1)^{\ell_1}(x-c_2)^{\ell_2}\ldots (x-c_s)^{\ell_s}g(x)
\]
where $c_1,c_2,\ldots c_s\in GR(p^a,\ell)$ and $g(x)$ has no linear factor in $GR(p^a,\ell)[x]$. Thus, all linear factors of $f(x)$ over $GR(p^a,\ell)$ has the form:
\[
x-\gamma,\text{ where } \gamma=c_i+h\cdot p^{\lceil \frac{a}{\ell_i}\rceil},h\in GR(p^a,\ell),1\leq i\leq s.
\]
 We then generalize the Guruswami-Sudan list decoding algorithm of Reed-Solomon codes to Galois rings. Consider a Reed-Solomon code $\cC\subseteq GR(p^a,\ell)^n$ of length $n$ and dimension $k$. Let $\{\alpha_1,\alpha_2,\ldots,\alpha_n\}$ be the evaluation set, let $e$ denote the number of error positions, and let $(y_1,y_2,\ldots,y_n)\in GR(p^a,\ell)^n$ be the received word. Our list decoding algorithm first find a non-zero polynomial $Q(X,Y)$ with $(1,k-1)$ degree at most $n-k$, such that $Q(\alpha_i,y_i)=0$ with multiplicity $r$ for every $1\leq i\leq n$. Next, we factorize $Q(X,Y)$ with respect to $Y$ into linear factors $Y-f(X)$ and list $f(X)$ as the candidate codeword. We show that the above algorithm can efficiently list-decode Reed-Solomon codes up to the Johnson bound.

\paragraph{List Decoding Algorithm for FRS codes.}
We generalize the list decoding of FRS codes over finite field to that over Galois rings. Assume that $0\leq t\leq N$, $D\geq 1$ and the received word:
\[
\mathbf{y} =
\begin{pmatrix}
y_0 & y_m & \cdots & y_{n-m} \\
\vdots & \vdots & \ddots & \vdots \\
y_{m-1} & y_{2m-1} & \cdots & y_{n-1}
\end{pmatrix}
\in GR(p^a,\ell)^{m \times N}, \quad N = \frac{n}{m}
\]
We first compute non-zero polynomial $Q(X,Y_1,\ldots,Y_s)$ as follows:
\[
Q(X,Y_1,\ldots Y_s)=A_0(X)+A_1(X)Y_1+\ldots +A_s(X)Y_s,
\]
where deg$[A_0]\leq D+k-1$ and deg$[A_i]\leq D$ for every $1\leq i\leq s$, such that for all $0\leq i\leq N$ and $0\leq j\leq m-s$,
\[
Q(\gamma^{im+j}, y_{im+j}, \ldots, y_{im+j+s-1}) = 0.
\]
We compute $\ell$ such that $X^\ell$ is the largest common power of $X$ among $A_0(X),\ldots,A_s(X)$ and for every $0\leq i\leq s$, $A_i(X)\leftarrow \frac{A_i(X)}{X^\ell}$. We write $A_i(X)$ as $A_i(X)=\sum_{j=0}^{D+k-1}a_{ij}X^j$ for every $0\leq i\leq s$ and rewrite the equation:
 \begin{align*}
0 = C(X)&= Q\left(X, f(X), f(\gamma X), \ldots, f(\gamma^{s-1}X)\right)\\
&=\sum_{j=0}^{D+k-1}a_{0,j}X^j+\sum_{i=1}^s\bigg(\sum_{j=0}^{D}a_{i,j}X^j\bigg)\bigg(\sum_{j=0}^{k-1}f_j\gamma^{(i-1)j}X^j\bigg)
\end{align*}
If $p\nmid Q(X,Y_1,\ldots, Y_s)$, let $h$ be the largest integer such that $p$ divides the common divisor of $\{a_{i,j}:0\leq i\leq s, 0\leq j<h \}$. This means $p$ is not the common divisor of $a_{0,h},\ldots,a_{s,h}$ and let $B(X)=a_{1,h}+a_{2,h}X+\ldots+a_{s,h}X^{s-1}$. To find a suitable polynomial $f(x)$, we consider the solutions of the linear system formed by the coefficients of $X^r$ for $r\geq h$:
\[
a_{0,r}+\sum_{i=1}^{s}\bigg(f_i(\sum_{j=1}^s a_{j,r-i}\gamma^{(j-1)i})\bigg)=0.
\]
Since $B(X)$ has degree at most $s-1$, it has at most $s-1$ units of the form $\gamma^i$ as its roots. By fixing at most $s-1$ $f_i$'s, we can obtain a unique solution for the coefficients of $f(X)$ satisfying the required conditions. Moreover, the number of such assignments is at most $p^{a\ell (s-1)}$, which implies that all such coefficients of polynomials $f(X)$ lie in a free module of rank at most $s-1$.

Since $p \mid A_0(X), A_1(X),\ldots,A_s(X)$, this implies that the received codeword is zero when modulo $p$. Since the received codeword module $p$ corresponds to a valid codeword over the field $\F_{p^\ell}$, all candidate codewords become uniquely determined after this reduction, namely the zero codeword. Thus, we can divide the received codeword by $p$ and focus on the case of $GR(p^{a-1},\ell)$. The same argument can be applied to the case $p^{i} \mid A_0(X), A_1(X),\ldots,A_s(X)$
and then we conclude that the linear system of coefficients has at most $p^{(a-i)\ell (s-1)}$ solutions and these solutions lie in a free module of rank $s-1$.

\paragraph{Deterministic Pruning for Folded Reed--Solomon codes.}
The interpolation and root-finding procedure above shows that all candidate codewords lie in an affine free module $A=\mathbf{c}_0+U\subseteq \cC_a$ of rank at most $s-1$. It remains to prune this module and output exactly the codewords in the Hamming ball. We extend the deterministic pruning method of~\cite{ashvinkumar2026algorithmic} to Galois rings. The main idea is to work through the quotient rings $GR(p^t,\ell)$, $1\le t\le a$, layer by layer. Given a close codeword $\mathbf{c}_t$ modulo $p^t$, all of its lifts inside $A_{t+1}$ can be written uniquely as $\widehat{\mathbf{c}}_t+p^t \mathbf{w}^{(t+1)}$ for $\mathbf{w}\in U_1$ where $U_1=U \bmod p$ is a linear subspace over the residue field $\mathbb F_q$. For each partial codeword $\mathbf{c}_t$, we construct a residual received word over $\mathbb F_q$, apply the finite-field deterministic pruning algorithm to $U_1$, and then perform an explicit distance check after lifting back to $GR(p^{t+1},\ell)$. This gives a deterministic pruning algorithm that outputs exactly $A\cap L(\mathbf{y},\rho)$, and hence the complete decoding list.

For the Folded Reed--Solomon decoding radius $\rho_s=\frac{s}{s+1}\left(1-\frac{mR}{m-s+1}\right),$ using the list-size bound for the quotient codes, the number of active candidates at every layer is at most $s$. Therefore, the pruning algorithm runs in time
\[
O(as)\cdot \left(\frac{2}{\eta}\right)^{O(2^s)}Nm\cdot \operatorname{poly}(a,s,\log q),
\]
where $\eta=\delta-\rho_s$. In particular, by taking $s=\Theta(1/\varepsilon)$ and $m=\Theta(1/\varepsilon^2)$, we obtain exact deterministic pruning up to radius $1-R-\varepsilon$ in time
\[
a\cdot\left(\frac{1}{\varepsilon}\right)^{O(2^{O(1/\varepsilon)})} Nm\cdot \operatorname{poly}(a,\log q,1/\varepsilon).
\]

\paragraph{Improved List Size for Folded Reed-Solomon Codes.}
We improve our list size by extending the recent progress in folded RS codes to Galois ring. 
Although the state-of-the-art result about the list decoding of folded RS codes is due to ~\cite{chen2025explicit}, we do not know how to generalize their results to Galois ring. Instead, we prove a tighter bound on the list size $O(\frac{1}{\epsilon^2})$ by extending the approach in $~\cite{srivastava2025improved}$. 
Let $\mathcal{H}$ be a free module of $GR(p^a,\ell)[X]^{<Rn}$ with rank $s$, i.e. there exists polynomials $h_0,h_1,\ldots,h_s$ such that
\[
\mathcal{H}=\left\{h_0+\sum_{j=1}^{s}\alpha_{j}h_{j}:\forall j\in [s],\alpha_{j}\in GR(p^a,\ell) \right\},
\]
where the set of polynomials $\{h_1,h_2,\ldots,h_s\}$ is linearly independent over $GR(p^a,\ell)$. The condition that a polynomial $h=h_0+\sum_{j=1}^{s}\alpha_jh_j$ agrees with any polynomial $y$ on position $i\in [N]$ after folding can be written as a linear system:
\[
\left[
\begin{array}{cccc}
h_1(\gamma^{(i-1)m}) & h_2(\gamma^{(i-1)m}) & \cdots & h_s(\gamma^{(i-1)m}) \\
h_1(\gamma^{(i-1)m+1}) & h_2(\gamma^{(i-1)m+1}) & \cdots & h_s(\gamma^{(i-1)m+1}) \\
\vdots & \vdots & \ddots & \vdots \\
h_1(\gamma^{(i-1)m+m-1}) & h_2(\gamma^{(i-1)m+m-1}) & \cdots & h_s(\gamma^{(i-1)m+m-1})
\end{array}
\right]
\left[
\begin{array}{c}
\alpha_1 \\
\alpha_2 \\
\vdots \\
\alpha_s
\end{array}
\right]
=
\left[
\begin{array}{c}
(y - h_0)(\gamma^{(i-1)m}) \\
(y - h_0)(\gamma^{(i-1)m+1}) \\
\vdots \\
(y - h_0)(\gamma^{(i-1)m+m-1})
\end{array}
\right]
\]
Let us call the $m\times s$ matrix appearing above as $A_i$ for $i\in [N]$, and denote $r_i=rank_M(A_i)$. Using the equivalence condition for the linear independence of the polynomials $f_1,\ldots f_s\in GR(p^a,\ell)[X]$ over $GR(p^a,\ell)$, together with the constraint on the number of roots, it follows that the rank of the matrix satisfies a certain inequality:
\[
\sum_{i=1}^{N}(s-r_i)\leq \frac{s\cdot Rn}{m-s+1}.
\]
We denote $\mathcal{H}_y = \mathcal{H} \cap \mathcal{L} \left( \vec{y}, \frac{b}{b+1} \cdot \left( 1 - \frac{m}{m-b+1} \cdot R \right) \right)$, and $S_h$ be the agreement set between $y$ and $h$ (over all of $[N]$). Utilizing the lower bound on the size of agreement sets,
\[
\left( \frac{1}{b+1} + \frac{bR}{b +1} \cdot \frac{m}{m-b+1} \right) N |\mathcal{H}_y| \leq \sum_{h \in \mathcal{H}_y} |S_h|.
\]
The above rank inequality yields the upper bound of $\sum_{h \in \mathcal{H}_y} |S_h|$:
\[
\sum_{h \in \mathcal{H}_y} |S_h|\leq |E| \cdot |\mathcal{H}_y| + N \left( 1 - e + (b - 1)s \left( \frac{m}{m - s + 1} R - e \right) \right)
\]
Thus, by simplifying these inequalities we conclude that:
\[
|\mathcal{H}_y|<(b-1)s+1.
\]

We further show that the optimal finite-field list-size phenomenon for folded Reed--Solomon codes also extends to Galois rings at the combinatorial level. More precisely, for every integer $1 \le L \le m$, we prove that the $m$-folded Reed--Solomon code over $\mathrm{GR}(p^{a},\ell)$ is
\[
\left(\frac{L}{L+1}\left(1-\frac{mR}{m-L+1}\right),L\right)
\]
average-radius list-decodable. Our proofs start from a bad list of size $L+1$. We define a global agreement weight $W$ and  first derive the lower bound $W \ge \frac{Lk}{m-L+1}$ from the failure of average-radius list-decodability. We then organize the bad list  according to its $p$-adic congruence structure by constructing a rooted congruence tree. 

For each effective internal node $C$ of this tree, we pass to a residual folded Reed--Solomon instance over the residue field $\mathbb{F}_{p^\ell}$ and associate to it a local weight $W(C)$. Using the finite-field local contradiction lemma node-by-node, we show that
\[
W(C) < \frac{(|\mathrm{ch}(C)|-1)k}{m-|\mathrm{ch}(C)|+2}.
\]
where $\mathrm{ch}(C)$ denote the chilren of node $C$. A tree decomposition identity then yields $W=\sum_{C\in \mathrm{Int}(T)} W(C)$, while a tree-threshold inequality gives
\[
\sum_{C\in \mathrm{Int}(T)}
\frac{|\mathrm{ch}(C)|-1}{m-|\mathrm{ch}(C)|+2}
\le
\frac{L}{m-L+1}.
\]
Combining these estimates leads to the global upper bound $W<\frac{Lk}{m-L+1}$, contradicting the initial lower bound. This shows that no bad list exists. In particular, our argument is not a straightforward extension of~\cite{chen2025explicit} from finite fields to Galois rings, but instead relies on an additional tree-based analysis tailored to the $p$-adic congruence structure of Galois rings.

We emphasize that the above $O(1/\varepsilon^2)$ bound is algorithmic: it is obtained from the linear-algebraic decoding framework and the free-module structure of the candidate list. In contrast, the optimal $O(1/\varepsilon)$ bound proved later is combinatorial and is established in the average-radius sense.

\subsection{Organizations}

In this paper, we first provide a brief review of the fundamental concepts of Galois rings and coding theory in Section 2, and then present in Section 3 some results concerning the solution of linear equations over Galois rings. In Section 4, by exploiting the property of unique quasi-prime factorization in Galois rings, we generalize the method in ~\cite{guruswami1998improved} to Galois rings, thereby enabling list decoding for codes of rate $r$ up to $1-\sqrt{r}$ fraction of errors. In Section~\ref{section:5}, we generalize the list decoding framework in ~\cite{guruswami2011linear} to FRS codes over Galois rings and prove a list of polynomial size. In Section~\ref{sec:galois-frs-deterministic-pruning}, we adapt the deterministic pruning method of~\cite{ashvinkumar2026algorithmic} to Galois rings by lifting candidates through the quotient rings $GR(p^t,\ell)$ and reducing each lifting step to a finite-field pruning problem over the residue field. In section~\ref{section:6}, inspired by the approach in ~\cite{srivastava2025improved}, we develop a refined analysis that yields an algorithmic list-size bound $O(\frac{1}{\varepsilon^2})$ by bounding the intersection of the code with a free module. In section~\ref{sec:optimal-combinatorial-list-size}, we further establish an optimal combinatorial average-radius list size theorem for folded Reed--Solomon codes over Galois rings using the results in~\cite{chen2025explicit}. 

\section{Preliminaries}
\subsection{Galois Ring}
Galois ring is a finite ring with identity $1$ such that the set of its zero divisors with $0$ added forms a principal ideal $(p\cdot 1)$ for some prime number $p$. 
Let $a,\ell\geq 1$, $h(x)$ be a monic basic irreducible polynomial of degree $\ell$ in $\Z_{p^a}[x]$, then the residue class of ring $GR(p^a,\ell) = \Z_{p^a}[x]/(h(x))$ is a Galois ring and $\F$ its residue field $\F_{p^\ell}$. Let $GR(p^a,\ell)[X]$ be the polynomial ring over $GR(p^a,\ell)$. $GR(p^a,\ell)[x]^{<k}$ is the collection of polynomials of degree less than $k$ in $GR(p^a,\ell)[x]$. We denote by $GR(p^a, \ell)^n$  the collection of vectors of length $n$ over $GR(p^a, \ell)$ and $GR(p^a, \ell)^{n\times m}$ the collection of $n\times m$ matrices over  $GR(p^a, \ell)$.

There exists a nonzero element $\gamma$ in Galois ring $GR(p^a,\ell)$ such that $1, \gamma,\ldots,\gamma^{p^\ell-2}$ consists of the roots of $x^{p^\ell-1}-1$ in $GR(p^a,\ell)$. 
The irreducible polynomial $h(x)$ with root $\gamma$ is called a basic primitive polynomial in $\mathbb{Z}_{p^a}[x]$. 
There are two ways to represent an element in $GR(p^a, \ell)$. 
For  $c\in GR(p^a,\ell)$, One can write $c=c_0+c_1\gamma +...+c_{\ell-1}\gamma^{\ell-1}$ where $c_0,c_1,...,c_{\ell-1}\in \mathbb{Z}_{p^a} \}.$ 
On the other hand, we can also represent an element $c\in GR(p^a,\ell)$ as $c=b_0+b_1p+...+b_{a-1}p^{a-1}$, where $b_0,b_1,...,b_{a-1}\in \{0,1,\gamma,\gamma^2,...,\gamma^{p^\ell-2} \}$. From this representation, $c$ is a unit if and only if $b_0\neq 0$. This also implies that any element in $GR(p^a,\ell)$ is either an unit or divisible by $p$.

Similar to the extension field, we can also define the extension of Galois ring. 
\begin{lemma}[Theorem~14.23 \cite{wan2011finite}]
    Let $h(x)$ be a basic irreducible polynomial of degree $\ell$ over $R=GR(p^s,m)$. Then the residue class ring $R[x]/(h(x))$ is a Galois ring of characteristic $p^s$ and cardinality $p^{sm\ell}$ and contains $R$ as a subring. Thus
\[
R[x]/(h(x)) = GR(p^s,m\ell).
\]
\label{residuering}
\end{lemma}

To factorize a polynomial $f$ over Galois ring, we need to first factorize $f$ in the residue field and then apply Hensel lifting lemma to find its factor over Galois Ring. The following two lemmas state this fact.
\begin{lemma}[Lemma~14.20 \cite{wan2011finite}]\label{lm:coprime factor}
    Let $R=GR(p^a,\ell)$ and $f$ be a monic polynomial in $R[x]$ and $g_1, g_2, \ldots, g_r$ be pairwise coprime monic polynomials in $\overline{R}[x]$. Assume that $\overline{f} = g_1 g_2 \ldots g_r$ in $\overline{R}[x]$. Then there exist pairwise coprime monic polynomials $f_1, f_2, \ldots, f_r$ in $R[x]$ such that $f = f_1 f_2 \ldots f_r$ and $\overline{f_i} = g_i$ for $i = 1, 2, \ldots, r$.\label{Hensel}
\end{lemma}
\begin{lemma}[Hensel Lemma \cite{wan2011finite}]
    Let $f$ be a monic polynomial of degree $\geq 1$ in $R[x]$. Then
\begin{enumerate}
    \item[(i)] $f$ can be factorized into a product of some number, say $r$, of pairwise coprime monic primary polynomials $f_1, f_2, \ldots, f_r$ over $R$:
    \[
    f = f_1 f_2 \cdots f_r
    \]
    and for each $i = 1, 2, \ldots, r$ $\overline{f_i}$ is a power of a monic irreducible polynomial over $\F_{p^\ell}$.
    \item[(ii)] Let
    \[
    f = f_1 \cdots f_r = h_1 \cdots h_t
    \]
    be two factorizations of $f$ into products of pairwise coprime monic primary polynomials over $R$, then $r = t$ and after renumbering, $f_i = h_i$, $i = 1, 2, \ldots, r$.
\end{enumerate}
\label{factor}
\end{lemma}
Let $f(x)\in GR(p^a,l)[x]$ and $\F_{p^\ell}$ is the residue field of the Galois ring. If $p\nmid f(x)$, then $\overline{f(x)}$ is not a zero polynomial, and its reduction $\bmod p$ can be factored in $\F_{p^\ell}[x]$ as $\overline{f(x)}=(x-a_1)^{\ell_1}(x-a_2)^{\ell_2}\ldots (x-a_s)^{\ell_s}\overline{g(x)} \bmod p$, where $g(x)$ has no linear factor over $\F_{p^\ell}$. Applying the \cref{Hensel}, we can lift each factor $(x-a_i)$ and $\overline{g(x)}$ to $GR(p^a,l)[x]$ as:
\[
f(x)=(x-c_1)^{\ell_1}(x-c_2)^{\ell_2}\ldots (x-c_s)^{\ell_s}g(x)
\]
where $c_1,c_2,\ldots,c_s\in GR(p^a,\ell)$ and $g(x)$ has no linear factor in $GR(p^a,\ell)[x]$.

\begin{thm}\label{thm:linear factor}
    All linear factors of $f(x)$ over $GR(p^a,\ell)$ take the form: 
    \[
    x-\gamma,  \text{  where  }  \gamma=c_i+h\cdot p^{\lceil\frac{a}{l_i}\rceil},  h\in GR(p^a,\ell), 1\leq i\leq s.
    \]
\end{thm}

\begin{proof}
   Write $c_i\in GR(p^a,\ell)$ as $c_i=c_{i,0}+c_{i,1}p+\ldots c_{i,a-1}p^{a-1}$, $1\leq i\leq s$, and write $\gamma\in GR(p^a,\ell)$ as $\gamma=\gamma_0+\gamma_{1}p+\ldots+\gamma_{a-1}p^{a-1}$. By \cref{factor}, the factorization $f(x)=\prod_{i=1}^{s}(x-c_i)^{l_i}\cdot g(x)$ is the unique factorization of $f(x)$ into primary components. Moreover, by \Cref{lm:coprime factor}, the residue classes $c_{i,0}\in \F_{p^\ell}$ are pairwise distinct for $i\neq j$. Now suppose $\gamma\in GR(p^a,\ell)$ is a root of $f(x)$, i.e. $f(\gamma)=0$. We note that $g(\gamma)$ is not divisible by $p$ or otherwise $\overline{g(x)}=g(x) \bmod p$ has a root $\gamma \bmod p$. This implies that $g(x)$ has a linear factor and the contradiction happens. Since $g(\gamma)$ is an unit in $GR(p^a,\ell)$ and $f(\gamma)=0$, this implies that at least one of $(\gamma-c_i)$ is divisible by $p$. The fact that $c_{i,0}$ are all distinct leads to the conclusion that there exists a unique index $i\in [s]$ such that $\gamma\equiv c_i\ \bmod p$. That is,
    \[
    (\gamma-c_i)^{\ell_i}=0\ \textit{ and } \ (\gamma-c_j)^{\ell_j}\neq 0\ \textit{ for }\ j\neq i.
    \]
    In a Galois ring, the condition $(\gamma-c_i)^{\ell_i}=0$ implies that 
    \[
    \gamma=c_i+h\cdot p^{\lceil\frac{a}{\ell_i}\rceil}
    \]
    for some $y\in GR(p^a,\ell)$ with the standard structure theory of nilpotent roots. Thus, we complete the proof.
\end{proof}

\cref{thm:linear factor} yields a corollary on the factorization of polynomials over Galois rings into linear factors.
\begin{col}
    Let $f(x)\in GR(p^a,l)[X]$ and $f(x)=p^m(x-c_1)^{\ell_1}(x-c_2)^{\ell_2}\ldots (x-c_s)^{\ell_s}g(x)$, where $c_1,c_2,\ldots,c_s\in GR(p^a,\ell)$, $p\nmid g(x)$ and $g(x)$ has no linear factor in $GR(p^a,\ell)[x]$. Then all linear factors of $f(x)$ over $GR(p^a,l)$ take the form:
    \[
    x-\gamma,  \text{  where  }  \gamma=c_i+h\cdot p^{\lceil\frac{a-m}{l_i}\rceil},  h\in GR(p^a,\ell), 1\leq i\leq s.
    \]
\end{col}

\subsection{Codes over Galois Ring} 
Let $\Sigma$ be a finite alphabet\footnote{In our application, $\Sigma$ can be either Galois ring or finite field.} and let $\mathbf{x}$ and $\mathbf{y} \in \Sigma^n$. Then the Hamming distance between them is $d(\mathbf x,\mathbf y) := |\{i \in [n]:x_i\neq y_i\}|$. Given a vector $\bx$ and a subset $\mathcal{Y} \subseteq \Sigma^n$ we denote $d(\bx,\mathcal{Y}) := \min\{d(\bx,\by):\by \in \mathcal{Y}\}$. The Hamming distance of code $\cC$ is $d(\cC)= \min_{\bx,\by\in \cC,\bx\neq \by}d(\bx,\by)$. The relative distance of $\cC$ is $\delta=\frac{d(\cC)}{n}$ and the rate is $r=\frac{\log_{|\Sigma|}|\cC|}{n}$. 
Let $\cC$ be a code with length $n$, Hamming distance $d$ and rate $r$ over alphabet $\Sigma$. Given $\bv\in \Sigma^n$, we use $\mathcal{L}(\bv,d)$ to denote the list of codeswords in $\cC$ whose distance from $\bv$ is lest than $d$. That is, $\mathcal{L}(\bv,d)=\{\bc\in \cC: d(\bv,\bc)<d \}$. We say that a code is combinatorially list decodable up to radius $d$ if for every $\bv\in \Sigma^n$, $\mathcal{L}(\bv,d)$ is of size at most polynomial in $n$. Likewise, we say a code can be efficiently list decodable up to radius $d$ if it is combinatorially list decodable up to $d$, and the list $\mathcal{L}(\bv,d)$ can be found in polynomial time in $n$. We will also use the notion of average-radius list-decodability. Let $L\ge 1$ be an integer, we say that $\cC$ is $(\rho,L)$ average-radius list-decodable if for every set of pairwise distinct codewords $\bc_{0},\bc_{1},\ldots,\bc_{L}\in \cC$ and every word $\by\in \Sigma^{N}$, we have $\frac{1}{L+1}\sum_{i=0}^{L} d(\bc_{i},\by) > \rho n$.

\paragraph{Reed-Solomon Codes over Galois Ring.} 
Assume that $GR(p^a,\ell)$ is a Galois ring with $p^\ell-2\geq n$ and $\gamma\in GR(p^a,\ell)$ of multiplicative order $p^{\ell}-1$. We can similarly generalize the celebrated Reed-Solomon code to its counterpart over Galois ring $GR(p^a,\ell)$. Given a polynomial $f(X)$ of degree at most $k$, the encoding algorithm $Enc_{RS}$ of Reed-Solomon codes is 
\[
f(x) \rightarrow \left(
f(1), f(\gamma)\ldots, f(\gamma^{n-1})
\right) \in GR(p^a,\ell)^n.
\]
The code $\cC_{RS}$ is denoted by $\cC_{RS}=\{Enc_{RS}(f(x)):f(x)\in GR(p^a,l)[x]^{<k}\}$. One can show that this Reed-Solomon code has code length $n$, rate $\frac{k}{n}$ and minimum distance $n-k+1$.

\paragraph{Folded Reed-Solomon Codes over Galois Ring.} 
One can also generalize the folded Reed-Solomon codes to its counterpart over Galois ring $GR(p^a,\ell)$. Given a polynomial $f(X)$ of degree at most $k$, the encoding algorithm $Enc_{FRS}$ of the $m$-folded Reed-Solomon code is 
\[
f(x) \rightarrow \left[ 
\begin{pmatrix}
f(1) \\
f(\gamma) \\
\vdots \\
f(\gamma^{m-1})
\end{pmatrix}, 
\begin{pmatrix}
f(\gamma^m) \\
f(\gamma^{m+1}) \\
\vdots \\
f(\gamma^{2m-1}) \\
\end{pmatrix}, 
\ldots, 
\begin{pmatrix}
f(\gamma^{n-m}) \\
f(\gamma^{n-m+1}) \\
\vdots \\
f(\gamma^{n-1}) \\
\end{pmatrix} 
\right] \in \left(GR(p^a,\ell)^m\right)^\frac{n}{m}
\]
One can show that this Reed-Solomon code has code length $\frac{n}{m}$, rate $\frac{k}{n}$ and minimum distance $N-\lceil \frac{k}{m}\rceil+1$. The code $\cC_{FRS}$ is denoted as $\cC_{FRS}=\{Enc_{FRS}(f(x)):f(x)\in GR(p^a,\ell)[x]^{<k}\}$.


\section{Solving Linear Equations over Galois Rings}
In this section, we study the solvability of linear equations over Galois rings. Specifically, we consider the equation:
\begin{equation}\label{eq:matrix equation}
A\bx=\mathbf{b}
\end{equation}
where $A=(a_{ij})_{n\times m}$ is a matrix over $GR(p^a,\ell)$, and $\bx,\mathbf{b}$ are column vectors in $GR(p^a,\ell)$. Since the underlying ring is not a field, we leverage the notion of McCoy rank to generalize the classical concept of matrix rank.

\begin{defn}[McCoy Rank]
Let $R$ be a non-trivial commutative ring with identity, and let $A=(a_{ij})_{n\times m}$ be a matrix over $R$. If every entry $a_{ij}$ has a non-zero annihilator, then $rank_MA$ is defined to be zero. Otherwise, the rank of $A$ is the greatest positive integer $r\leq n$ such that the determinant of all $r\times r$ submatrices of $A$ does not have a common non-zero annihilator. Denote this rank by $rank_MA=r$.
\end{defn}

\begin{lemma}[Lemma~I.26 \cite{mcdonald2020linear}]
    Let $P,Q$ be invertible matrices over $R$ of appropriate dimensions, then we have: $rank_M(PAQ)=rank_M(A)$.
\end{lemma}

We now show how to adapt Gaussian elimination to compute the McCoy rank of a matrix over $GR(p^a,\ell)$. Suppose a matrix $A\in GR(p^a, \ell)^{n\times m}$. If every entry in $A$ is either zero or zero-divisors, then $rank_M(A)=0$. Otherwise, assume without of generality that the leading entry $a_{11}$ is a unit. Then, We can apply Gaussian elimination to eliminate $a_{1i}$ for $2\leq i\leq m$ and $a_{j1}$ for $2\leq j\leq n$ and thereby reduce the matrix to the block-triangular form:

\[
A \sim
\begin{pmatrix}
a_{11} & 0 \\
0 & A'
\end{pmatrix}
\]

One can continue this process until we obtain the following form:

\[
A \sim
\begin{pmatrix}
    P & 0 \\
    0 & Q
\end{pmatrix}
\]

where $P=diag\{p_{11},p_{22},\ldots,p_{rr}\}$ is a diagonal matrix with invertible diagonal entries $p_{ii}\in GR(p^a,l)^{\times}$, and $Q$ is a matrix in which all entries are either zero or zero-divisors. Then, we conclude $rank_M(A)=r$. It is clearly this process can be done in polynomial time in $m,n$.

\begin{lemma}[Theorem~51 \cite{mccoy1948rings}]
    When $b=0$, the system of equations \Cref{eq:matrix equation} has a nontrivial solution if and only if the rank of the matrix $A$ is less than the length of $\bx$.\label{solve equations}
\end{lemma}

\begin{col}
    When $b=0$, since the McCoy rank $rank_M(A)\leq min\{n,m \}$, if $n<m$, the system of equations \Cref{eq:matrix equation} always has a nontrivial solution. In addition, the proof in  \cite{mccoy1948rings} is constructive and thus produces the non-trivial solution in polynomial time.
    \label{col_solution}
\end{col}

\begin{defn}
    Let $N(R)$ denote the subset of commutative ring $R$ consisting of all elements which are not zero-divisors. Then $N(R)$ is a multiplicative subset of $R$ which contains the units of $R$.
    We say that $S\subseteq N(R)$ is subtractive in $N(R)$ if for all distinct $a,b\in S,\ a-b\in N(R)$.
\end{defn}

\begin{lemma}[Lemma~2.1 \cite{norton2000key}]
    If $A$ is a square matrix over $R$ and $det(A)\in N(R)$, then the linear system $Ax=0$ has only the trivial solution.
\end{lemma}

\begin{lemma}[Proposition~2.4 \cite{norton2000key}]
    Let $R$ be a finite ring. Then every element of $N(R)$ is a unit.
    \label{subtractive}
\end{lemma}

By the \cref{subtractive}, the size of a subtractive set over $GR(p^a,l)$ is at most $p^l-1$, i.e., all nonzero elements of its residue field $\F_{p^l}$. As shown in $\cite{norton2000key}$, by elementary row operationsthe, the parity-check matrix $H$ of Reed-Solomon codes over $GR(p^a,l)$ takes the following triangular form:
\[
\begin{bmatrix}
1 & 1 & 1 & \cdots & 1 \\
0 & \alpha_2 - \alpha_1 & \alpha_3 - \alpha_1 & \cdots & \alpha_n - \alpha_1 \\
0 & 0 & \prod_{i=1}^{2} (\alpha_3 - \alpha_i) & \cdots & \prod_{i=1}^{2} (\alpha_n - \alpha_i) \\
\vdots & \vdots & \vdots & \ddots & \vdots & \\
0 & 0 & 0 &\cdots & \prod_{i=1}^{k-1} (\alpha_n - \alpha_i)
\end{bmatrix}
\]

The determinant of any $d-1$ columns is a product of several expressions of the form $(\alpha_i-\alpha_j),i\neq j$, requiring the evaluation set of the $RS$ code is subtractive set and thus all its elements $\alpha_i,1\leq i\leq n$ are pairwise distinct non-zero elements from $\F_{p^l}$.

\begin{thm}[Theorem~3.3 \cite{norton2000key}]
    We denote the RS code defined by the above parity-check matrix over $GR(p^a,l)$ by $RS_{k}(\alpha)$, with $k=n-d+1$, all elements in the evaluation set $\alpha$ are nonzero elements of $\F_{p^l}$. $RS_k(\alpha)$ is a free $R-module$ of rank $k$, its minimum distance is $d$.
\end{thm}

\begin{col}
    The RS code defined by the above parity-check matrix satisfies $n=k+d-1$, and is an MDS code.
\end{col}

\section{List Decoding Algorithm of Reed-Solomon Codes}
In this section, we extend the classical list decoding algorithm of Reed–Solomon codes, originally proposed in~\cite{guruswami1998improved}, to the context of codes defined over Galois rings. 

If $p\mid f(X)\in GR(p^a,\ell)[X]$, then we assume $f(X)=p^{i}g(X)$, where $p\nmid g(X)$. 
According to \Cref{thm:linear factor}, 
we can perform a linear factorization of the polynomial $f(X)=(X-c_1)^{\ell_1}(X-c_2)^{\ell_2}\ldots(X-c_s)^{\ell_s}g(X)$, where $g(x)$ has no linear factor. Analogous to above analysis, by replacing $p$ with $p^i$, we can obtain all linear factors of the form:
\[
X-\gamma, \quad \gamma=c_j+h\cdot p^{\lceil\frac{a-i}{l_j}\rceil},1\leq j\leq s.
\]
Consequently, we summarize this result in the following theorem.
\begin{thm}
Let $f(X) \in \mathrm{GR}(p^a, \ell)[X]$.

\begin{itemize}
    \item If $p \nmid f(X)$ and $f(X)$ admits the factorization
    \[
    f(X) = \prod_{i=1}^{s}(x - c_i)^{\ell_i} \cdot g(X),
    \]
    where $c_i \in \mathrm{GR}(p^a, \ell)$ and $g(X)$ has no linear factor in $\mathrm{GR}(p^a, \ell)[X]$, then all linear factors of $f(x)$ are of the form
    \[
    X - \gamma, \quad \text{where } \gamma = c_i + h \cdot p^{\left\lceil \frac{a}{\ell_i} \right\rceil}, \quad h \in \mathrm{GR}(p^a, \ell).
    \]

    \item If $p^i \parallel f(X)$ and $f(X)$ can be written as
    \[
    f(X) = p^i \cdot \prod_{j=1}^{s}(X - c_j)^{\ell_j} \cdot g(x),
    \]
    where $g(X)$ has no linear factor and $c_j \in \mathrm{GR}(p^a, \ell)$, then all linear factors of $f(X)$ are of the form
    \[
    X - \gamma, \quad \text{where } \gamma = c_j + h \cdot p^{\left\lceil \frac{a - i}{\ell_j} \right\rceil}, \quad h \in \mathrm{GR}(p^a, \ell).
    \]
\end{itemize}
\end{thm}
This procedure is formalized as \cref{Algorithm1}: Linear Factorization of Polynomials over Galois Rings.

\begin{algorithm}[H]
\caption{Linear Factorization of polynomial $f(x)$}
\label{Algorithm1}
\begin{algorithmic}[1]
\Statex \textbf{Input:} $f(X)\in GR(p^a,\ell)[X]$
\Statex \textbf{Output:} All linear factors $X-\gamma$ of $f(X)$

\State $f(X)\leftarrow \frac{f(X)}{p^j}$, where $j$ is the largest integer such that $p^j\mid f(X)$
\State $\overline{f(X)}=(X-\overline{c_1})^{\ell_1}(X-\overline{c_2})^{\ell_2}\ldots(X-\overline{c_s})^{\ell_s}\overline{g(X)}\leftarrow$ linear factorization of $\overline{f(X)}$
\State $f(X)=(X-c_1)^{\ell_1}(X-c_2)^{\ell_2}\ldots(X-c_s)^{\ell_s}g(X)\leftarrow$ Hensel Lift of $\overline{f(X)}$
\State Output $X-\gamma,\gamma=c_i+y\cdot p^{\lceil\frac{a-j}{l_i}\rceil},1\leq i\leq s$
\end{algorithmic}
\end{algorithm}

Adapted from~\cite{guruswami1998improved}, we now present a list decoding algorithm \cref{Algorithm2}. for Reed–Solomon codes over $GR(p^a,\ell)$.
\begin{algorithm}
\caption{The List Decoding Algorithm for Reed-Solomon Codes over Galois Ring}
\label{Algorithm2}
\begin{algorithmic}[1]
\Statex \textbf{Input:} $n\geq k\geq 1,d\geq 1,r\geq 1,e=n-t$ and $n$ pairs $\{(\alpha_i,y_i)\}_{i=1}^{n}$
\Statex \textbf{Output:} List of polynomials $f(X)$ of degree at most $k-1$.

\State Find a non-zero $Q(X,Y)$ with $(1,k-1)$ degree at most $d$, such that:
\[
Q(\alpha_i,y_i)=0
\]
with multiplicity $r$ for every $1\leq i\leq n$. 
\State $\mathcal{L}\leftarrow \emptyset$
\For{every factor $Y-f(X)$ of $Q(X,Y)$}
\If{$d(y_i,(f(\alpha_i))_{i=1}^n)\leq e$ and $\deg(f)\leq k-1$}
\State Add $f(X)$ to $\mathcal{L}$
\EndIf
\EndFor
\end{algorithmic}
\end{algorithm}

\begin{rmk}
    The key distinction from the field case lies in the factorization step, for which we developed \cref{Algorithm1} to identify all linear factors over $GR(p^a,\ell)$.
\end{rmk}

In the description given in Section 3, the evaluation points are required to be pairwise distinct elements from $\F_{p^l}$. As a result, the following two lemmas hold.

\begin{lemma}[Lemma 5~\cite{guruswami1998improved}]
    The first step of \cref{Algorithm2} imposes $\binom{r+1}{2}$ constraints for each $i$ on the coefficients of $Q(X,Y)$.
\end{lemma}

\begin{lemma}[Lemma~3~\cite{guruswami1998improved}]
    $R(X):=Q(X, f(X))$ has $r$ roots for every $i$ such that $f(\alpha_i) = y_i$. In other words, $(X - \alpha_i)^r$ divides $R(X)$.
    \label{R(X)}
\end{lemma}

\begin{rmk}
    The proofs of the aforementioned lemmas do not rely on the properties of fields and can be naturally generalized to Galois ring. By \cref{residuering}, We can generalize the factorization method presented in~\cite{lidl1997finite} to Galois rings. Next, we introduce the following lemma to reveal the relationship between the degree of a polynomial and the number of its roots over a Galois ring.
\end{rmk}

\begin{lemma} \label{rootslemma}
    Let $f(X)\in GR(p^a,l)[X]$ be a non-zero polynomial with $deg(f)\leq t$, then $f(X)$ has at most $t$ units as roots (counting multiplicities).
\end{lemma}

\begin{proof}
    Express $f(X)$ as
    $$f(X)=f_0(X)+f_1(X)p+\ldots +f_{a-1}(X)p^{a-1}, f_i(X)=f_{i0}+f_{i1}X+\ldots +f_{it}X^{t}, f_{ij}\in \F_{p^\ell}.$$
    Suppose that the polynomial $f(X)$ has more than $t$ units as roots (counting multiplicities). If $f_0(X)\neq 0$, then $f_0(X)=0\, \bmod p$ has more than $t$ roots (counting multiplicities) in $\F_{p^\ell}$ and a contradiction happens.
    If $f_0(X)\neq 0$, we assume that $f_0(X),f_1(X),\ldots f_{i-1}(X)= 0, f_i(X)\neq 0$, then we obtain:
    \[
    p^if_i(X)=0 \bmod p^{i+1}
    \]
    Hence, $f_i(X)$ has more than $t$ roots (counting multiplicities) in $\F_{p^\ell}$ and a contradiction happens. This completes the proof.
\end{proof}

\begin{lemma}\label{lm:roots}
        Let $f(X)\in GR(p^a,\ell)$ be a non-zero polynomial with degree $t$, $f(X)=f_0(X)+f_1(X)p+\ldots +f_{a-1}(X)p^{a-1}, f_i(X)=f_{i0}+f_{i1}X+\ldots +f_{it}X^{t}, f_{ij}(X)\in \F_{p^\ell} 
 (0\leq i\leq a-1, 0\leq j\leq t)$. If $f_0(X)\neq 0$, then there are at most $t$ units $\alpha_1,\ldots,\alpha_t$ such that $p\mid f(\alpha_i)$. (counting multiplicities)
\end{lemma}

The following theorem is a direct consequence of \cref{lm:roots}.
\begin{thm}
    Let $Q(X, Y)$ be computed by Step 1 in \cref{Algorithm2}. Let $f(X)$ be a polynomial of degree $\leq k-1$ such that $f(\alpha_i) = y_i$ for at least $t > \frac{d}{r}$ many values of $i$. Then, $Y - P(X)$ divides $Q(X, Y)$.
\end{thm}

\begin{proof}
    By \cref{R(X)}, $\alpha_i$ is a root of $R(X)$ with multiplicity $r$. From \cref{rootslemma}, we know that if $tr>d$, then $R(X)\neq 0$. This completes the proof.
\end{proof}

Now, we are ready to present the main result of this section, the list decoding algorithm of Reed-Solomon codes over Galois ring up to Johnson bound.
\begin{thm}
    \cref{Algorithm2} can efficiently list decode Reed-Solomon codes of rate $r$ up to $1 - \sqrt{r}$ fraction of errors.
\end{thm}

\section{List Decoding Algorithm of Folded Reed-Solomon Codes}\label{section:5}
In this section, we extend the list decoding framework to FRS codes defined over Galois rings. Our algorithm is adapted from the general framework in~\cite{guruswami2011linear}. We now describe a list decoding algorithm tailored to folded RS codes over Galois rings. The algorithm follows a two-step structure analogous to~\cite{guruswami2011linear}. We briefly summarize \cref{Algorithm3} as follows.

\textbf{Step 1:} Interpolate a non-zero multivariate polynomial $Q(X,Y_1,\ldots,Y_s)$, where each variable $Y_i$ has degree one such that
\begin{equation}\label{eq:4}
Q(\gamma^{im+j},y_{im+j},\ldots,y_{im+j+s-1})=0
\end{equation}
for all $0\leq i<N$ and $0\leq j\leq m-s$.

\textbf{Step 2:} Identify all polynomials $f(X)\in GR(p^a,\ell)[X]$ such that
\[
Q(X,f(X),f(\gamma X)\ldots,f(\gamma^{s-1}X))= 0
\]
and $f(X)$ agrees with the received word on at least $t$ folded positions. 

\begin{algorithm}
\caption{The List Decoding Algorithm for Folded Reed-Solomon Codes over Galois Ring}
\label{Algorithm3}
\begin{algorithmic}[1]
\Statex \textbf{Input:} An agreement parameter $0 \leq t \leq N$, parameter $D \geq 1$ and the received word:
\[
\mathbf{y} =
\begin{pmatrix}
y_0 & y_m & \cdots & y_{n-m} \\
\vdots & \vdots & \ddots & \vdots \\
y_{m-1} & y_{2m-1} & \cdots & y_{n-1}
\end{pmatrix}
\in GR(p^a,\ell)^{m \times N}, \quad N = \frac{n}{m}
\]

\Statex \textbf{Output:} All polynomials $f(X) \in GR(p^a,\ell)[X]$ of degree at most $k-1$ such that for at least $t$ values of $0 \leq i < N$,
\begin{equation}\label{eq:2}
\begin{pmatrix}
f(\gamma^{mi}) \\
\vdots \\
f(\gamma^{m(i+1)-1})
\end{pmatrix}
=
\begin{pmatrix}
y_{mi} \\
\vdots \\
y_{m(i+1)-1}
\end{pmatrix}
\end{equation}

\State Compute non-zero polynomial $Q(X, Y_1, \ldots, Y_s)$ as follows:
\[
Q(X, Y_1, \ldots, Y_s) = A_0(X) + A_1(X)Y_1 + A_2(X)Y_2 + \ldots + A_s(X)Y_s,
\]
where $\deg[A_0] \leq D + k - 1$ and $\deg[A_i] \leq D$ for every $1 \leq i \leq s$, such that for all $0 \leq i < N$ and $0 \leq j \leq m - s$,
\[
Q(\gamma^{im+j}, y_{im+j}, \ldots, y_{im+j+s-1}) = 0
\]

\State $\mathcal{L} \gets \emptyset$

\For{every $f(X) \in GR(p^a,\ell)[X]$ such that 
\begin{equation}\label{eq:3}
Q\left(X, f(X), f(\gamma X), \ldots, f(\gamma^{s-1}X)\right)=0 
\end{equation}}

  \If{$\deg(f) \leq k-1$ and $f(X)$ satisfies \Cref{eq:2} for at least $t$ values of $i$}
    \State Add $f(X)$ to $\mathcal{L}$
  \EndIf
\EndFor

\State \textbf{return} $\mathcal{L}$

\end{algorithmic}
\end{algorithm}

Next, we analyze the correctness of \cref{Algorithm3}. We begin with the result showing that there exists a nonzero polynomial $Q(X,Y_1,\ldots,Y_s)$ for \textbf{Step 1}.

\begin{lemma}
    If $D\geq \lfloor\frac{N(m-s+1)-k+1}{s+1}\rfloor$, then there exists a non-zero polynomial $Q(X,Y_1,\ldots,Y_s)$ that satisfies \textbf{Step 1} of \cref{Algorithm3}. \label{chooseD}
\end{lemma}

\begin{proof}
    All coefficients in $A_i(X)$ are the variables. Thus, the number of variables is
    \[
    D+k+s(D+1)=(s+1)(D+1)+k-1
    \]
    On the other hand, the number of constraints in \Cref{eq:4} is $N(m-s+1)$.
    Note that if the variables outnumber the equations, by \cref{solve equations}, there exists a non-zero $Q$ that satisfies \textbf{Step 1}. This means 
    \[
    (s+1)(D+1)+k-1>N(m-s+1)
    \]
    which can be reduced to
    \[
    D>\frac{N(m-s+1)-k+1}{s+1}-1.
    \]
    This is guaranteed by the condition of this lemma.
\end{proof}

\begin{lemma}
    If $t>\frac{D+k-1}{m-s+1}$, then every polynomial $f(X)$ in the output list $\mathcal{L}$ satisfies \Cref{eq:3}.  \label{chooset}
\end{lemma}

\begin{proof}
    Consider the polynomial $R(X)=Q(X,f(X),f(\gamma X),\ldots,f(\gamma^{s-1}X))$. Because the degree of $f(\gamma^l X)$ is at most $k-1$. This implies $\deg(R)\leq D+k-1$.
    Let $f(X)\in \mathcal{L}$ be one of the polynomials of degree at most $k-1$. Assume that $f(X)$ agrees with the received word at column $i$ for some $0\leq i<N$, i.e.,
    \[
\begin{pmatrix}
f(\gamma^{mi}) \\
\vdots \\
f(\gamma^{m(i+1)-1})
\end{pmatrix}
=
\begin{pmatrix}
y_{mi} \\
\vdots \\
y_{m(i+1)-1}
\end{pmatrix}
\]
Then, for all $0\leq j\leq m-s$, we have
\begin{eqnarray*}
&&R(\gamma^{mi+j})=Q(\gamma^{mi+j},f(\gamma^{mi+j}),f(\gamma^{mi+1+j}),\ldots,f(\gamma^{mi+s-1+j}))\\
&&=Q(\gamma^{mi+j},y_{mi+j},y_{mi+1+j},\ldots,y_{mi+s-1+j})=0.
\end{eqnarray*}
Note that for all $0\leq i<N, 0\leq j\leq m-s$, $\gamma^{mi+j}$ is a unit in $GR(p^a,\ell)$. Thus, the number of roots in  $R(X)$ as a unit is at least
\[
t(m-s+1)>D+k-1\geq deg(R).
\]
By the \cref{rootslemma}, this implies that $R(X)=0$ and thus $f(X)$ satisfies \Cref{eq:3} as desired.
\end{proof}

The major challenge in generalizing the list decoding algorithm to Galois rings lies in solving the root-finding equation from \emph{Step 2}. This is because the standard linear algebra over fields can not apply directly. To address this challenge, we propose an iterative recursive strategy which is presented in \cref{Algorithm4}.

\begin{thm}
    Let $Q(X,Y_1,\ldots,Y_s)$ be a non-zero multivariate polynomial, every $f(X)\in GR(p^a,\ell)[X]$ satisfies that \Cref{eq:3} is found by the \Cref{Algorithm4}.
\end{thm}
\begin{proof}
Let $f(X)=\sum_{i=0}^{k-1}f_iX^i$ satisfies \Cref{eq:3}. \Cref{Algorithm4} will output the coefficient of $f(X)$ one by one. Let $Q_i(X,Y_1,\ldots,Y_s)$ and
$M_i(X,Y_1,\ldots,Y_s)=X^{-r_i}Q_i(X,Y_1,\ldots,Y_s)$ be the  $M(X,Y_1,\ldots,Y_s)$ and $Q(X,Y_1,\ldots,Y_s)$ in the $i$-th iteration of the "for" loop in \Cref{Algorithm4}. 
Note that it holds
\[
Q_{i+1}(X,Y_1,\ldots,Y_s)=X^{-r_i}M_i(X,XY_1+f_i,\gamma XY_2+f_i,...,\gamma^{s-1}XY_s+f_i).
\]
Since $X$ does not divide $M_i(X,Y_1,\ldots,Y_s)$, it holds $M_i(0,Y_1,\ldots,Y_s)\neq 0$. Let $g_j(X)=\sum_{i=j}^{k-1}f_iX^{i-j}$. We prove that 
$Q_i(X,g_i(X),g_i(\gamma X),\ldots,g_i(\gamma^{s-1}X))=0$ by induction on $i$, where the induction base $i = 0$ is obvious as $g_0(X)=f(X)$ satisfying \Cref{eq:3}. Assume that this holds for $i=j$, i.e., $Q_j(X,g_j(X),g_j(\gamma X),\ldots,g_j(\gamma^{s-1}X))=0$. Then, for $i=j+1$, we observe that $Xg_{j+1}(X)+f_j=g_j(X)$. This means
\[
\begin{aligned}
&Q_{j+1}(X,g_{j+1}(X),g_{j+1}(\gamma X),\ldots,g_{j+1}(\gamma^{s-1}X))\\
&=M_j(X,Xg_{j+1}(X)+f_j,\gamma Xg_{j+1}(\gamma X)+f_j,\ldots,\gamma^{s-1}Xg_{j+1}(\gamma^{s-1}X)+f_j)\\
&=X^{-r_j}Q_j(X,g_j(X),g_j(\gamma X),\ldots,g_j(\gamma^{s-1}X))=0.
\end{aligned}
\]
We complete the induction.
Then $Q_i(X,g_i(X),g_i(\gamma X),\ldots,g_i(\gamma^{s-1}X))=0$ implies 
\[
M_i(X,g_i(X),g_i(\gamma X),\ldots,g_i(\gamma^{s-1}X))=0.
\]
We set $X=0$ to obtain $M_i(0,f_i,\ldots,f_i)=M_i(0,g_i(0),\ldots,g_i(0))=0$.
Thus, in the $i$-th iteration of the "for" loop in  \Cref{Algorithm4}, the algorithm will output the coefficient $f_i$. We complete the proof.
\end{proof}

\begin{algorithm}
\caption{Find all $f(X)\in GR(p^a,\ell)[X]$ satisfies that \Cref{eq:3}.}
\label{Algorithm4}
\begin{algorithmic}[1]
\Statex \textbf{Input:} $(Q(X, Y_1, \ldots, Y_s), k, i)$ with 
\[
Q(X, Y_1, \ldots, Y_s) = A_0(X) + A_1(X)Y_1 + A_2(X)Y_2 + \ldots + A_s(X)Y_s,
\]
where $\deg[A_0] \leq D + k - 1$ and $\deg[A_i] \leq D$ for every $1 \leq i \leq s$, such that for all $0 \leq i < N$ and $0 \leq j \leq m - s$,
\[
Q(\gamma^{im+j}, y_{im+j}, \ldots, y_{im+j+s-1}) = 0
\]

\Statex \textbf{Output:} All polynomials $f(X) \in GR(p^a,\ell)[X]$ such that
\[
Q\left(X, f(X), f(\gamma X), \ldots, f(\gamma^{s-1}X)\right)= 0.
\]

\State Find the largest integer $r$ for which $Q(X,Y_1,\ldots, Y_s)/X^r$ is still a polynomial.
\State $M(X,Y_1,\ldots,Y_s)\leftarrow Q(X,Y_1,\ldots, Y_s)/X^r$
\State Find all roots of the polynomial $M(0,Y,Y,\ldots,Y)$
\For{each of the distinct roots $\zeta$ of $M(0,Y,Y,\ldots,Y)$}
\State $f_i\leftarrow \zeta$
\If{$i = k - 1$}
    \State output $f(X)=f_0+f_1X+\ldots+f_{k-1}X^{k-1}$
\Else
    \State $Q'(X,Y_1,\ldots,Y_s)=M(X,XY_1+\zeta,\gamma XY_2+\zeta,\ldots,\gamma^{s-1}XY_s+\zeta)$
    \State  Run \cref{Algorithm4} with input $(Q'(X,Y_1,\ldots,Y_s),k,i+1)$. 
  \EndIf
\EndFor

\end{algorithmic}
\end{algorithm}

Now, we analyze the error correcting capability of the algorithm. To satisfy the constraint in \cref{chooseD}, we pick
\[
D= \lfloor\frac{N(m-s+1)-k+1}{s+1}\rfloor
\]
This along with the constraint in \cref{chooset}, implies that the algorithm works as long as $t>\frac{D+k-1}{m-s+1}$.
The above is satisfied if we choose
\[
t>\frac{\frac{N(m-s+1)-k+1}{s+1}+k-1}{m-s+1}=\frac{N(m-s+1)+s(k-1)}{(s+1)(m-s+1)}.
\]
Thus, we would be fine if we pick
\[
t>N\left(\frac{1}{s+1}+\left(\frac{s}{s+1}\right)\left(\frac{m}{m-s+1}\right)\cdot R\right)
\]

\begin{thm}
    \cref{Algorithm4} can list decode Folded Reed-Solomon code with folding parameter $m\geq 1$ and rate $R$ up to $\frac{s}{s+1}\left(1-\frac{mR}{m-s+1}\right)$ fraction of errors.
\end{thm}

To show that our list decoding algorithm runs in polynomial time, we need to bound the output size of root finding algorithm. We next show that the number of solutions in the root finding step is bounded and moreover all the solutions lie within a free module.

\begin{thm}
Using the notation defined above, we consider two cases:
\begin{itemize}
    \item If $p\nmid Q(X_1,X_2,\ldots,X_s)$, there are at most $p^{a\ell(s-1)}$ solutions of $f(X)$ to the equations 
    \begin{equation}\label{eq:4}
    A_0(X)+A_1(X)f(X)+A_2(X)f(\gamma X)+\ldots A_s(X)f(\gamma^{s-1}X)= 0
    \end{equation}
    and all the solutions lie in a $GR(p^{a},\ell)$ free module.
    \item Otherwise $p^{i} \mid Q(X_1,X_2,\ldots,X_s)$ and $p^{i+1}\nmid Q(X_1,X_2,\ldots,X_s)$, there are at most $p^{(a-i)\ell(s-1)}$ solutions of $f(X)$ to the equations 
    \[
    A_0(X)+A_1(X)f(X)+A_2(X)f(\gamma X)+\ldots A_s(X)f(\gamma^{s-1}X)= 0
    \]
    and all the solutions lie in a $GR(p^{a-i},\ell)$-linear free module.
\end{itemize}\label{freemodule}
\end{thm}

\begin{proof}
    If $X$ are the common divisor of polynomials $A_0,A_1,\ldots,A_s$, we essentially just factor out the largest common power of $X$ from all of the $A_i's$, and proceed with the resulting polynomial. Let $l\geq 0$ be the largest integer such that $A_i(X)=X^l A'_i(X)$ for $0\leq i\leq s$; then $X$ does not divide all of $A'_i(X)$ and we have:
    \[
    X^l(A_0'(X)+A_1'(X)f(X)+\ldots+A_s'(X)f(\gamma^{s-1}X))= 0
    \]
    
    Then, we can apply the same argument by replacing $A_i(X)$ with $A'_i(X)$ since $A'_i(X)$ also satisfies \Cref{eq:4}.
    Hence, we now assume that $X$ is not the common divisor of $A_i(X)$. This implies that there exists some $h>0$ such that the constant term of the polynomial $A_{h}(X)$ is non-zero. We write $A_i(X)$ as
    \[
    A_i(X)=\sum_{j=0}^{D+k-1}a_{ij}X^j.
    \]
    for every $0\leq i\leq s$.
    We begin by considering the case in which $p$ does not divide $Q(X,Y_1,\ldots,Y_s)$, i.e. the g.c.d of $a_{0,0},\ldots,a_{s,D+k-1}$ is $1$. Then, we have
    \begin{align*}
0 = C(X) &= Q\left(X, f(X), f(\gamma X), \ldots, f(\gamma^{s-1}X)\right)= A_0(X) + A_1(X)f(X) + \cdots + A_s(X)f(\gamma^{s-1}X) \\
&= \sum_{j=0}^{D+k-1}a_{0,j}X^j+\sum_{i=1}^s\bigg(\sum_{j=0}^{D}a_{i,j}X^j\bigg)\bigg(\sum_{j=0}^{k-1}f_j\gamma^{(i-1)j}X^j\bigg).
\end{align*}
Let $h$ be the largest integer such that 
$p$ divides the common divisor of $\{a_{i,j}:0\leq i\leq s, 0\leq j<h \}$. This means $p$ is not the common divisor of $a_{0,h},\ldots,a_{s,h}$.
Since $C(X)=0$, each coefficient of $C(X)$ is zero. Now, we consider the coefficient of $X^r$ for $r\geq h$
\begin{equation}\label{eq:solution of f_i}
a_{0,r}+\sum_{i=1}^{s}\bigg(f_i(\sum_{j=1}^s a_{j,r-i}\gamma^{(j-1)i})\bigg)=0.
\end{equation}

\noindent
Let
$$
B(X)=a_{1,h}+a_{2,h}X+\ldots+a_{s,h}X^{s-1}.
$$
Notice that there exists an element such that $p\nmid a_{j,h},j\in [s]$, so $B(X)$ is non-zero polynomial. By \cref{lm:roots}, there are at most $s-1$ distinct $\gamma^m$ for $0\leq m\leq k-1$ such that $p\mid B(\gamma^m)$. Without loss of generality, we assume that $p\nmid B(\gamma^m)$ for $m=1,\ldots,k-s$. We fix $f_{k-s+1},\ldots,f_{k-1}$ to be any value in $GR(p^a,\ell)$. Then, we want to prove that once $f_{k-s+1},\ldots,f_{k-1}$ are fixed, $f_0,\ldots,f_{k-s}$ are unique. 
We write \Cref{eq:solution of f_i} as the linear equations $A(f_0,\ldots,f_{k-1})=0$ where $A$ is a $k\times (D+k-1-h)$ matrix. Let $D$ be the submatrix of $A$ by taking out the first $k-s$ columns. Then, we have 
\[
D=(d_{i,j}):=
\begin{pmatrix}
    B(1) & * & * & \ldots & * \\
    \circ & B(\gamma) & \ldots & \ldots & * \\
    \vdots & \vdots & \vdots & \ddots & \vdots \\
    \circ & \circ &  & \ldots & B(\gamma^{k-j-1}) \\
    \ldots & \ldots &\ldots &\ldots &\ldots
\end{pmatrix}
\]
where all the main diagonal elements $d_{i,i}(1\leq i\leq k-j)$ are $B(\gamma^{i-1})$, and the element $*$ in the upper right corner of the matrix $D$ are divided by $p$ as they can be represented as the linear combination of the elements $a_{i,j}$ for $0\leq i\leq s, 0\leq j\leq h-1$. This means the submatrix  $D$ has full rank.

By applying Gaussian elimination to eliminate the upper right corner elements of the matrix $D$, we can obtain at most a unique solution for $(f_0,f_1,\ldots,f_{k-j-1})$.

Based on the above analysis, it can be concluded that when $p$ does not divide $A_0(X),A_1(X)\\ ,\ldots,A_s(X)$, the system of equations has at most $p^{a\ell (s-1)}$ solutions. Consequently, the size of the list is at most $p^{a\ell(s-1)}$ and all solutions lie in a free module of rank $s-1$.

If $p \mid A_0(X), A_1(X),\ldots,A_s(X)$, it implies that the received corrupted codeword reduces to zero modulo $p$. Since the reduction module $p$ still corresponds to a valid codeword over the field $\F_{p^\ell}$, all candidate codewords become uniquely determined after reduction, namely the zero codeword. Thus, we can claim that all candidate codeword is divisible by $p$ and thus we replace $Q$ with $\frac{Q}{p}$ and invoke the list decoding algorithm over $GR(p^{a-1},\ell)$ instead. 
Since $p^i \mid Q(X_1,X_2,\ldots,X_s)$ and $p^{i+1} \nmid Q(X_1,X_2,\ldots,X_s)$, we can then apply the analysis of case $1$ over $GR(p^{a-i},\ell)$ to obtain the candidate codewords $\bc_1,\ldots,\bc_a$ over $GR(p^{a-i},\ell)$. Then, the real candidate codewords are $p^i\cdot \bc_1,\ldots,p^i\cdot \bc_r$. It is clear that \Cref{eq:4} has at most $p^{(a-i)\ell (s-1)}$ solutions and these solutions lie in a $GR(p^{a-i},\ell)$-linear free module of rank $s-1$.

The above procedure can be iterated $i$ times, thereby applying the analysis of case 1 over $GR(p^{a-i},\ell)$. Therefore, \Cref{eq:4} has at most $p^{(a-i)\ell (s-1)}$ solutions and these solutions lie in a free module of rank $s-1$.

\end{proof}

Motivated by the above theorem and the ideas used in its proof, we develop the following~\cref{Algorithm5} to support the process of \Cref{eq:3}.

\begin{algorithm}
\caption{Another Method for Find All $f(X)\in GR(p^a,\ell)[X]$ satisfies that Eq.(3)}
\label{Algorithm5}
\begin{algorithmic}[1]
\Statex \textbf{Input:} $A_0(X),\ldots,A_s(X)$
\Statex \textbf{Output:} All polynomials $f(X) \in GR(p^a,\ell)[X]$ such that
\[
Q\left(X, f(X), f(\gamma X), \ldots, f(\gamma^{s-1}X)\right) = 0 
\]

\State Compute $l$ such that $X^l$ is the largest common power of $X$ among $A_0(X),\ldots,A_s(X)$.
\For{every $0\leq i\leq s$}
\State $A_i(X)\leftarrow \frac{A_i(X)}{X^l}$.
\EndFor
\If{$p\nmid A_0(X),\ldots,A_s(X)$}
\State Find the smallest $i$ such that $p\nmid a_{0,i},a_{1,i},\ldots a_{s,i}(0\leq i\leq D)$ and compute $B(X)$.
\State Compute all the integer $j$ such that $B(\gamma^j)(0\leq j\leq k-1)\in (p1)$.
\State Each coefficient $f_j$ is assigned a value in $GR(p^a,\ell)$.
\State $(f_0,f_1,\ldots, f_{k-1})\leftarrow$ solve each equation similar to \Cref{eq:solution of f_i}.
\EndIf
\If{$p \mid A_0(X),\ldots,A_s(X)$}
\State Find the integer $i$ such that $p^i\mid A_0(X),\ldots,A_s(X)$ and $p^{i+1}\nmid A_0(X),\ldots,A_s(X)$.
\State $A_0(X),\ldots,A_s(X)\leftarrow \frac{A_0(X)}{p^i},\ldots,\frac{A_s(X)}{p^i}$.
\State Find the smallest $i'$ such that $p\nmid a_{0,i'},a_{1,i'},\ldots a_{s,i'}(0\leq i'\leq D)$ and compute $B(X)$.
\State Compute all the integer $j$ such that $B(\gamma^j)(0\leq j\leq k-1)\in (p1)$.
\State Each coefficient $f_j$ is assigned a value in $GR(p^{a-i},\ell)$.
\State $(f_0,f_1,\ldots, f_{k-1})\leftarrow$ solve each equation similar to \Cref{eq:solution of f_i} over $GR(p^{a-i},\ell)$.
\EndIf

\end{algorithmic}
\end{algorithm}

In this section, we have generalized the list decoding algorithm for folded Reed–Solomon codes to the Galois ring setting. We established an explicit error-correction capability and provided a detailed analysis of the root-finding step. In particular, we proved that the number of valid decoded polynomials is bounded and forms a free module, thereby ensuring algorithmic feasibility and list size control.

\section{Algorithmic Improvements to List Decoding of Folded Reed--Solomon Codes}\label{sec:galois-frs-deterministic-pruning}
In Section~\ref{section:5}, the interpolation and root-finding procedure for folded Reed--Solomon codes over Galois rings gives a low-rank free module containing all candidate codewords. In this section, we explain how to prune this free module deterministically and output exactly the codewords in the Hamming ball.

Let $\F_q:=GR(p,\ell)$ with $ q=p^\ell$. For $1\le t\le a$, we identify $GR(p^t,\ell)\cong GR(p^a,\ell)/(p^t).$ Thus, reducing a vector, codeword, or polynomial modulo $p^t$ simply means reducing each coordinate modulo $p^t$. The Teichm\"uller unit $\gamma\in GR(p^a,\ell)$ used for evaluation will also denote its image in every quotient ring $GR(p^t,\ell)$.

Let $\cC_a\subseteq (GR(p^a,\ell)^m)^N$ be the $m$-folded Reed--Solomon code of block length $N=n/m$ and rate $R=k/n$ from Section~\ref{section:5}. For $1\le t\le a$, write $\cC_t:=\cC_a \pmod{p^t}\subseteq (GR(p^t,\ell)^m)^N.$ Then $\cC_t$ is the $m$-folded Reed--Solomon code over $GR(p^t,\ell)$ with the same parameters. For $\mathbf z\in (GR(p^t,\ell)^m)^N$ and radius $\rho$, define $\mathcal L_t(\mathbf z,\rho):=\{\mathbf c\in \cC_t: d(\mathbf c,\mathbf z)\le \rho\}.$ When $t=a$, we simply write $\mathcal L(\mathbf z,\rho)$ for $\mathcal L_a(\mathbf z,\rho)$.

For a finite $GR(p^a,\ell)$-module $H$, let $\mu(H):=\min\{s:\ H \text{ can be generated by } s \text{ elements over } GR(p^a,\ell)\}$. By Nakayama's lemma, or equivalently Proposition~2.8 of~\cite{atiyah2018introduction}, for $GR(p^a,\ell)$ with residue field $\F_q$, $\mu(H)=\dim_{\F_q} H/pH$.

\begin{lemma}[Smith normal form over Galois rings]\label{lem:smith-normal-form}
Let $M$ be a finite free $GR(p^a,\ell)$-module. If $H\subseteq M$ is an $GR(p^a,\ell)$-submodule, then there exists an $GR(p^a,\ell)$-basis $\mathbf{e}_1,\ldots,\mathbf{e}_K$ of $M$, an integer $s\ge 0$, and integers $0\le b_1\le b_2\le\cdots\le b_s<a$ such that $H=\bigoplus_{i=1}^s p^{b_i}GR(p^a,\ell) \mathbf{e}_i$. Moreover, $s=\mu(H)$.
\end{lemma}

\begin{proof}
The decomposition follows from the Smith normal form over Galois rings~\cite[Theorem~2.5]{guo2015several}. It remains to prove $s=\mu(H)$. Reducing modulo $pH$, we obtain $H/pH\cong \bigoplus_{i=1}^s p^{b_i}GR(p^a,\ell)/p^{b_i+1}GR(p^a,\ell)$. Each quotient $p^{b_i}GR(p^a,\ell)/p^{b_i+1}GR(p^a,\ell)$ is one-dimensional over the residue field $\F_q=GR(p^a,\ell)/(p)$. Hence, $\mu(H)=\dim_{\F_q}H/pH=s$.
\end{proof}

We first recall the deterministic pruning theorem over finite fields.

\begin{thm}[Corollary~\cite{ashvinkumar2026algorithmic}]\label{thm:field-deterministic-pruning}
Let $\cC\subseteq (\F_q^m)^N$ be an $\F_q$-linear code of relative distance $\delta$. Given a received word $\mathbf z\in(\F_q^m)^N$, a parameter $0<\eta<\delta$, and an $h$-dimensional affine subspace $H\subseteq \cC$, there is a deterministic algorithm $\mathrm{DPrune}_{\F_q}(\mathbf z,H,\eta)$ that outputs $H\cap \{\mathbf c\in \cC: d(\mathbf c,\mathbf z)\le \delta-\eta\}.$ Its running time is $\left(\frac{2}{\eta}\right)^{O(2^h)} Nm\cdot \mathrm{polylog}(q).$

In the folded Reed--Solomon application over finite fields,~\cite[Corollary~3.6]{ashvinkumar2026algorithmic} shows that, after the linear-algebraic list-decoding step, all codewords within radius $1-R-\varepsilon$ are contained in an affine subspace of dimension $h=O(1/\varepsilon)$, while the distance gap satisfies $\eta=\Theta(\varepsilon)$. Hence, applying $\mathrm{DPrune}_{\F_q}$ outputs the complete finite-field list in deterministic time $\left(\frac{2}{\varepsilon}\right)^{O(2^{1/\varepsilon})} n\cdot \mathrm{polylog}(n)$ with $n=Nm.$
\end{thm}

By the root-finding analysis in Section~\ref{section:5}, for the received word $\mathbf y$ and the radius $\rho$, the first-stage decoder produces an affine free module $A=\mathbf f+U\subseteq \cC_a$ such that $\mathcal L(\mathbf y,\rho)\subseteq A.$ Here $U=GR(p^a,\ell)\mathbf v_1\oplus\cdots\oplus GR(p^a,\ell)\mathbf v_h$ is a free $GR(p^a,\ell)$-module of rank $h\le s-1.$ For $1\le t\le a$, define $A_t:=A\pmod{p^t}$ and $U_t:=U\pmod{p^t}.$ Also write $\mathbf y^{(t)}:=\mathbf y\pmod{p^t},\,\mathbf v_j^{(t)}:=\mathbf v_j\pmod{p^t}.$ In particular, $U_1=\F_q\mathbf v_1^{(1)}\oplus\cdots\oplus \F_q\mathbf v_h^{(1)}$ is an $h$-dimensional $\F_q$-linear subspace of $\cC_1$. For $w\in U_1$, write it uniquely as $\mathbf{w}=\sum_{j=1}^h \lambda_j \mathbf{v}_j^{(1)}$ where $\lambda_j\in \mathbb F_q.$ Using the standard representative of $\mathbb F_q$ in each quotient
$GR(p^{t+1},\ell)$, define $\mathbf{w}^{(t+1)}:=\sum_{j=1}^h \lambda_j \mathbf{v}_j^{(t+1)}\in U_{t+1}$ where $\lambda_j$ is represented as $\lambda_j+0\cdot p+\cdots+0\cdot p^t$.

\begin{lemma}\label{lem:one-step-lifting}
Let $1\le t<a$. Fix $\mathbf c_t\in A_t$, and let $\widehat{\mathbf c}_t\in A_{t+1}$ be any lift of $\mathbf c_t$ modulo $p^{t+1}$,i.e., $\widehat{\mathbf c}_t\equiv \mathbf c_t \pmod{p^t}.$ Then all lifts of $\mathbf c_t$ inside $A_{t+1}$ are exactly $\widehat{\mathbf c}_t+p^t\mathbf w^{(t+1)}$ where $\mathbf w\in U_1.$ Moreover, the corresponding $\mathbf w\in U_1$ is unique and the number of
all lifts of $\mathbf c_t$ is $|U_1|=q^h$.
\end{lemma}

\begin{proof}
Let $\mathbf c_{t+1}\in A_{t+1}$ be another lift of $\mathbf c_t$. Since both $\mathbf c_{t+1}$ and $\widehat{\mathbf c}_t$ lie in the same affine module $A_{t+1}$, their difference belongs to $U_{t+1}$. Therefore, $\mathbf c_{t+1}-\widehat{\mathbf c}_t=\sum_{j=1}^{h}\mu_j\mathbf v_j^{(t+1)}$ for some $\mu_j\in GR(p^{t+1},\ell)$. The condition tha $\mathbf c_{t+1}\equiv \widehat{\mathbf c}_t \pmod{p^t}$ implies that every coefficient $\mu_j$ is divisible by $p^t$. Hence, $\mathbf c_{t+1}-\widehat{\mathbf c}_t= p^t\sum_{j=1}^{h}\lambda_j\mathbf v_j^{(t+1)}$ for $\lambda_j=\mu_j/p^t\pmod{p}$. Letting $\mathbf w=\sum_{j=1}^{h}\lambda_j\mathbf v_j^{(1)}\in U_1$ gives $\mathbf c_{t+1} =\widehat{\mathbf c}_t+p^t\mathbf w^{(t+1)}.$ Conversely, every vector of this form clearly reduces to $\mathbf c_t$ modulo $p^t$.

Finally, uniqueness follows from the direct-sum representation of $U_{t+1}$. If $p^t\mathbf w^{(t+1)}=p^t\mathbf w'^{(t+1)}$ for two $\mathbf w,\mathbf w'\in U_1$, then writing $\mathbf w-\mathbf w'=\sum_{j=1}^h\lambda_j\mathbf v_j^{(1)}$ gives $p^t\lambda_j=0$ in $GR(p^{t+1},\ell)$ for every $j$. Hence, $\lambda_j=0$ in $\F_q$ for all $j$, and so $\mathbf w=\mathbf w'$.
\end{proof}

Fix $\mathbf c_t\in A_t$ and a lift $\widehat{\mathbf c}_t\in A_{t+1}$ of $\mathbf c_t$. Define $\mathbf z^{(t)}:= \mathbf z^{(t)}(\mathbf c_t,\widehat{\mathbf c}_t)\in(\F_q^m)^N$ coordinatewise as follows. For $0\le i\le N-1$, if $(\mathbf c_t)_i=(\mathbf y^{(t)})_i,$ then $(\mathbf y^{(t+1)})_i-(\widehat{\mathbf c}_t)_i\in p^tGR(p^{t+1},\ell)^m$ and we let
\[
    \mathbf z_i^{(t)}:= \left(\frac{(\mathbf y^{(t+1)})_i-(\widehat{\mathbf c}_t)_i}{p^t}\right)\bmod p \in \F_q^m.
\]
If $(\mathbf c_t)_i\ne(\mathbf y^{(t)})_i$, we define $\mathbf z_i^{(t)}:=\mathbf{0}\in\F_q^m.$

\begin{lemma} \label{lem:residual-consistency}
Let $1\le t<a$, and suppose $\mathbf c_{t+1} =\widehat{\mathbf c}_t+p^t\mathbf w^{(t+1)}\in A_{t+1}$ where $\mathbf w\in U_1.$ If a folded coordinate $i$ satisfies $(\mathbf c_{t+1})_i=(\mathbf y^{(t+1)})_i,$ then $\mathbf w_i=\mathbf z_i^{(t)}.$ Moreover, $d(\mathbf w,\mathbf z^{(t)}) \le d(\mathbf c_{t+1},\mathbf y^{(t+1)}).$
\end{lemma}

\begin{proof}
If $(\mathbf c_{t+1})_i=(\mathbf y^{(t+1)})_i$, then reducing modulo $p^t$ gives $(\mathbf c_t)_i=(\mathbf y^{(t)})_i.$ Thus, $\mathbf z_i^{(t)}$ is defined by the first case above. Moreover,
\[
(\mathbf y^{(t+1)})_i-(\widehat{\mathbf c}_t)_i= (\mathbf c_{t+1})_i-(\widehat{\mathbf c}_t)_i = p^t(\mathbf w^{(t+1)})_i.
\]
Dividing by $p^t$ and reducing modulo $p$, we obtain $\mathbf z_i^{(t)}=\left(\frac{(\mathbf y^{(t+1)})_i-(\widehat{\mathbf c}_t)_i}{p^t}\right)=\left(\frac{(\mathbf c_{t+1})_i-(\widehat{\mathbf c}_t)_i}{p^t}\right)=\mathbf w_i.$ Therefore, every coordinate on which $\mathbf c_{t+1}$ agrees with $\mathbf y^{(t+1)}$ is also a coordinate on which $\mathbf w$ agrees with $\mathbf z^{(t)}$. Thus, $d(\mathbf w,\mathbf z^{(t)})\le d(\mathbf c_{t+1},\mathbf y^{(t+1)}).$
\end{proof}

The following algorithm performs pruning layer by layer. The explicit distance check is
necessary because Lemma~\ref{lem:residual-consistency} gives only a necessary condition
for a lift to be close to the received word.

\begin{algorithm}[H]
\caption{Deterministic pruning over $GR(p^a,\ell)$}
\label{alg:layered-pruning}
\begin{algorithmic}[1]
\Statex \textbf{Input:} A received word
$\mathbf y\in (GR(p^a,\ell)^m)^N$, an affine free module
$A=\mathbf f+U\subseteq \cC_a$, and a radius $\rho<\delta$.
\Statex \textbf{Output:} $A\cap\mathcal L(\mathbf y,\rho)$.

\State Let $\eta:=\delta-\rho$. For every $1\le t\le a$, let $\mathbf y^{(t)}:=\mathbf y \pmod{p^t}$.
\State $P_1\leftarrow \mathrm{DPrune}_{\F_q}(\mathbf y^{(1)},A_1,\eta)$ where $\mathrm{DPrune}_{\F_q}$ is defined in Theorem~\ref{thm:field-deterministic-pruning}

\For{$t=1,2,\ldots,a-1$}
    \State $P_{t+1}\leftarrow \emptyset$.
    \For{each $\mathbf c_t\in P_t$}
        \State Choose a deterministic lift
        $\widehat{\mathbf c}_t\in A_{t+1}$ such that $\widehat{\mathbf c}_t\equiv \mathbf c_t \pmod{p^t}$ and construct the residual received word $\mathbf z^{(t)} = \mathbf z^{(t)}(\mathbf c_t,\widehat{\mathbf c}_t) \in (\F_q^m)^N.$
        \State $W_t(\mathbf c_t)\leftarrow
        \mathrm{DPrune}_{\F_q}(\mathbf z^{(t)},U_1,\eta)$.
        \For{each $\mathbf w\in W_t(\mathbf c_t)$}
            \State $\mathbf c_{t+1}\leftarrow
            \widehat{\mathbf c}_t+p^t\mathbf w^{(t+1)}$.
            \If{$d(\mathbf c_{t+1},\mathbf y^{(t+1)})\le \rho$}
                \State Add $\mathbf c_{t+1}$ to $P_{t+1}$.
            \EndIf
        \EndFor
    \EndFor
\EndFor

\State \textbf{return} $P_a$.
\end{algorithmic}
\end{algorithm}

\begin{thm}\label{thm:layered-pruning-correctness}
For every $1\le t\le a$, the set $P_t$ produced by Algorithm~\ref{alg:layered-pruning} satisfies $P_t=A_t\cap \mathcal L_t(\mathbf y^{(t)},\rho).$ In particular, $P_a=A\cap\mathcal L(\mathbf y,\rho)$ and the algorithm outputs the complete list $\mathcal L(\mathbf y,\rho)$.
\end{thm}

\begin{proof}
We prove the statement by induction on $t$. For $t=1$, the claim follows from Theorem~\ref{thm:field-deterministic-pruning}:
\[
    P_1=\mathrm{DPrune}_{\F_q}(\mathbf y_1,A_1,\eta) = A_1\cap\{\mathbf c\in \cC_1:d(\mathbf c,\mathbf y_1)\le \delta-\eta\}.
\]
Since $\eta=\delta-\rho$, this is exactly $P_1=A_1\cap \mathcal L_1(\mathbf y_1,\rho).$

Assume that the claim holds for some $t<a$. We prove it for $t+1$. First, let $\mathbf c_{t+1}\in A_{t+1}\cap \mathcal L_{t+1}(\mathbf y^{(t+1)},\rho).$ Define $\mathbf c_t:=\mathbf c_{t+1}\pmod{p^t}$. Reducing modulo $p^t$ cannot decrease the agreement set with the received word, so $d(\mathbf c_t,\mathbf y^{(t)})\le d(\mathbf c_{t+1},\mathbf y^{(t+1)}) \le \rho.$

By the induction hypothesis, $\mathbf c_t\in P_t$. When the algorithm processes $\mathbf c_t$, it chooses some lift $\widehat{\mathbf c}_t\in A_{t+1}$. By Lemma~\ref{lem:one-step-lifting}, there is a unique $\mathbf w\in U_1$ such that $\mathbf c_{t+1} = \widehat{\mathbf c}_t+p^t\mathbf w^{(t+1)}.$ By Lemma~\ref{lem:residual-consistency},
\[
    d(\mathbf w,\mathbf z^{(t)})\le d(\mathbf c_{t+1},\mathbf y^{(t+1)}) \le \rho = \delta-\eta.
\]
Therefore $\mathbf w$ is output by $\mathrm{DPrune}_{\F_q}(\mathbf z^{(t)},U_1,\eta)$. The algorithm then constructs the lift $\widehat{\mathbf{c}}_t+p^t \mathbf{w}^{(t+1)}=\mathbf{c}_{t+1}.$ Since this lift satisfies $d(\mathbf{c}_{t+1},\mathbf{y}^{(t+1)})\le \rho$, it passes the explicit distance check and is added to $P_{t+1}$. Thus no close lift is missed.

Conversely, every element added to $P_{t+1}$ has the form $\widehat{\mathbf c}_t+p^t\mathbf w^{(t+1)}$ with $\widehat{\mathbf c}_t\in A_{t+1}$ and $\mathbf w^{(t+1)}\in U_{t+1}$. Hence it lies in $A_{t+1}$. Moreover, it is added only after passing the test $d(\mathbf c_{t+1},\mathbf y^{(t+1)})\le \rho.$ Therefore, $P_{t+1}\subseteq A_{t+1}\cap \mathcal L_{t+1}(\mathbf y^{(t+1)},\rho).$ Together with the previous paragraph, this proves $P_{t+1}=A_{t+1}\cap \mathcal L_{t+1}(\mathbf y^{(t+1)},\rho).$ The induction is complete. Taking $t=a$ gives $P_a=A\cap\mathcal L(\mathbf y,\rho).$ Since the first-stage decoder in Section~\ref{section:5} guarantees $\mathcal L(\mathbf y,\rho)\subseteq A$, we obtain $P_a=\mathcal L(\mathbf y,\rho).$
\end{proof}

We now give the final deterministic pruning guarantee for folded Reed--Solomon codes
over Galois rings. We specialize the previous layered pruning algorithm to the FRS
decoding radius. Fix $1\le s\le m$ and define $\rho_s:=\frac{s}{s+1}\left(1-\frac{mR}{m-s+1}\right),\,\eta:=\delta-\rho_s,$ where $\delta$ is the relative distance of the finite-field code $\cC_1$.

The first-stage decoder of Section~\ref{section:5} outputs an affine free module $A=\mathbf f+U\subseteq \cC_a$ with $\mathrm{rank}(U)=h\le s-1,$ such that $\mathcal L(\mathbf y,\rho_s)\subseteq A.$  We will also use the combinatorial list-size bound proved in Section~\ref{sec:optimal-combinatorial-list-size}: for every $1\le t\le a$ and every received word $z\in (GR(p^t,\ell)^m)^N$, we have $|L_t(z,\rho_s)|\le s.$ Indeed, this follows from the $(\rho_s,s)$ average-radius list-decodability of the $m$-folded Reed--Solomon code over $GR(p^t,\ell)$.

\begin{thm}\label{thm:galois-ring-frs-deterministic-pruning}
Under the above assumptions, Algorithm~\ref{alg:layered-pruning} outputs the exact list $\mathcal L(\mathbf y,\rho_s)$ and its running time is $O(as)\cdot\left(\frac{2}{\eta}\right)^{O(2^s)} Nm\cdot \mathrm{poly}(a,s,\log q).$
\end{thm}

\begin{proof}
The correctness follows from Theorem~\ref{thm:layered-pruning-correctness}. Indeed, Algorithm~\ref{alg:layered-pruning} outputs $A\cap \mathcal L(\mathbf y,\rho_s).$ Since the first-stage decoder in Section~\ref{section:5} guarantees $\mathcal L(\mathbf y,\rho_s)\subseteq A,$ the output is exactly $\mathcal L(\mathbf y,\rho_s)$.

It remains to bound the running time. By Theorem~\ref{thm:layered-pruning-correctness}, at every layer $t$ we have $P_t=A_t\cap \mathcal L_t(\mathbf y^{(t)},\rho_s).$ By the list-size bound for the quotient code $\cC_t$, $|P_t|\le |\mathcal L_t(\mathbf y^{(t)},\rho_s)|\le s.$ Therefore the total number of finite-field pruning calls is at most $1+(a-1)s=O(as).$

Each call to $\mathrm{DPrune}_{\F_q}$ is made on an affine translate of an $h$-dimensional subspace of $\cC_1$, either $A_1$ or a translate of $U_1$. Since $h\le s-1$, Theorem~\ref{thm:field-deterministic-pruning} gives the per-call time
\[
\left(\frac{2}{\eta}\right)^{O(2^h)}Nm\cdot \mathrm{polylog}(q)\le \left(\frac{2}{\eta}\right)^{O(2^s)} Nm\cdot \mathrm{polylog}(q).
\]
The additional work is lower-order: constructing $\mathbf{z}^{(t)}$, forming the lift $\widehat{\mathbf{c}}_t+p^t \mathbf{w}^{(t+1)}$, and checking the distance to $\mathbf{y}^{(t+1)}$ only require scanning $N$ folded positions of length $m$, and therefore take $Nm\cdot\mathrm{poly}(a,s,\log q)$ time per candidate. Since $|P_t|\le s$ at every layer, these auxiliary costs are absorbed into the pruning calls. Hence the total running time is
\[
O(as)\cdot \left(\frac{2}{\eta}\right)^{O(2^s)} Nm\cdot \mathrm{poly}(a,s,\log q).
\]
\end{proof}

Finally, let us record the near-capacity parameter choice. A direct computation gives
\[
    1-R-\rho_s= 1-R-  \frac{s}{s+1} \left(1-\frac{mR}{m-s+1}\right) = \frac{1-R}{s+1}  + \frac{sR(s-1)}{(s+1)(m-s+1)}.
\]
Thus, choosing $s=\Theta(1/\varepsilon)$ and $m=\Theta(1/\varepsilon^2)$,  gives $\rho_s\ge 1-R-\varepsilon.$ Since the relative distance of $\cC_1$ satisfies $\delta=1-R+O(1/N),$ we have $\eta=\delta-\rho_s=\Theta(\varepsilon)$ for $N$ large enough. Hence the pruning step runs in time
\[
a\cdot \left(\frac{1}{\varepsilon}\right)^{O(2^{O(1/\varepsilon)})} Nm\cdot \operatorname{poly}\bigl(a,\log q,1/\varepsilon\bigr).
\]

\section{Improved Algorithmic List Size for Folded Reed--Solomon Codes}\label{section:6}
In the previous section, we presented a list decoding algorithm for FRS codes over Galois rings and derived a preliminary upper bound on the output list size. In this section, inspired by the approach in~\cite{srivastava2025improved}, we refine this \emph{algorithmic} analysis and improve the list-size bound to $O(1/\varepsilon^2)$. This improvement is still obtained within the decoding framework of Section~\ref{section:5}, by exploiting the free-module structure of the solution space. Our key insight is to leverage the module structure of the solution space and an inductive dimension-reduction argument. First, based on the analysis in the previous section, we can obtain a result analogous to Theorem~3.5 in~\cite{srivastava2025improved}, confining the solution set to a free module.

\begin{thm}
    Let $\mathcal{C}_{FRS}$ be an $m$-Folded Reed-Solomon code of blocklength $N=\frac{n}{m}$ and rate $R$. For any integer $b$, $1\leq b\leq m$, and for any $\vec{y} \in \left(GR(p^a,\ell)^m \right)^N$, there exists a free-module $\mathcal{H}$ of $GR(p^a,\ell)[X]^{<Rn}$ of rank $b-1$ such that
    \[
    \mathcal{L}\left(\vec{y},\frac{b}{b+1}(1-\frac{m}{m-b+1}R) \right) \subseteq Enc_{FRS}(\mathcal{H}).
    \]
\end{thm}

Building on the above result, we now quantify the size of the list restricted to the structured free module. The following lemma, adapted from \cite{guruswami2011linear}, establishes bounds on the number of codewords in the list that lie within a free module.

\begin{lemma}[Lemma~4.1  \cite{srivastava2025improved}]\label{s=1}
    Let $\mathcal{C}$ be a linear code of distance $d$ and blocklength $N$ over alphabet $\F_q^m$, and let $\mathcal{H}\subseteq \mathcal{C}$ be an affine subspace of dimension $1$. Then, for any $\vec{y}\in (\F_q^m)^N$ and integer $b\geq 1$,
    \[
    \Bigl|\mathcal{H}\cap \mathcal{L}(\vec{y},\frac{b}{b+1}d)\Bigr|\leq b.
    \]
\end{lemma}


\begin{rmk}
As the proof of the above lemma in \cite{srivastava2025improved} does not rely on any intrinsic properties of the field, it can be naturally extended to Galois rings. Hence, we do not elaborate on it in this work. We proceed by reformulating the techniques from Section 5 of \cite{srivastava2025improved} within the framework of Galois ring, which enables us to establish an upper bound on the decoding list size.
\end{rmk}

Let $\mathcal{H}$ be a free module of $GR(p^a,\ell)[X]^{<Rn}$ with rank $s$, so that there exist polynomials $h_0,h_1,\ldots,h_s$ such that 
\[
\mathcal{H}=\left\{h_0+\sum_{j=1}^{s}\alpha_{j}h_{j}:\forall j\in [s],\alpha_{j}\in GR(p^a,\ell) \right\}.
\]
Moreover, the set of polynomials $\{h_1,h_2,\ldots,h_s \}$ is linearly independent over $GR(p^a,\ell)$, it implies that $h_i\neq 0\pmod p$ for every $1\leq i\leq s$.

The condition that a polynomial $h=h_0+\sum_{j=1}^{s}\alpha_{j}h_{j}$ agrees with any polynomial $f$ on position $i\in [N]$ after folding can be written as the collection of $m$ equations:
\[
\forall j\in [m], \qquad h(\gamma^{(i-1)m+j-1})=f(\gamma^{(i-1)m+j-1})
\]
Writing as a linear system,
\[
\left[
\begin{array}{cccc}
h_1(\gamma^{(i-1)m}) & h_2(\gamma^{(i-1)m}) & \cdots & h_s(\gamma^{(i-1)m}) \\
h_1(\gamma^{(i-1)m+1}) & h_2(\gamma^{(i-1)m+1}) & \cdots & h_s(\gamma^{(i-1)m+1}) \\
\vdots & \vdots & \ddots & \vdots \\
h_1(\gamma^{(i-1)m+m-1}) & h_2(\gamma^{(i-1)m+m-1}) & \cdots & h_s(\gamma^{(i-1)m+m-1})
\end{array}
\right]
\left[
\begin{array}{c}
\alpha_1 \\
\alpha_2 \\
\vdots \\
\alpha_s
\end{array}
\right]
=
\left[
\begin{array}{c}
(f - h_0)(\gamma^{(i-1)m}) \\
(f - h_0)(\gamma^{(i-1)m+1}) \\
\vdots \\
(f - h_0)(\gamma^{(i-1)m+m-1})
\end{array}
\right]
\]

Let us call the $m\times s$ matrix appearing above as $A_i$ for $i\in [N]$, and denote $r_i=rank_M(A_i)$. The following scenario differs significantly from the case in a field; a relevant analysis is now provided.

\begin{lemma}\label{lm:equiv}
    Let $\F_{p^\ell}$ be the residue field of $GR(p^a,\ell)$. The polynomials $f_1,f_2,\ldots,f_s\in GR(p^a,\ell)[X]$ are linearly independent over $GR(p^a,\ell)$ if and only if $f_1,f_2,\ldots,f_s\pmod p\in \F_{p^\ell}[X]$ are linearly independent over $\F_{p^\ell}$.
\end{lemma}

\begin{proof}
Let $f_1,f_2,\ldots,f_s\in GR(p^a,\ell)[X]$. We assume that $f_1,f_2,\ldots,f_s\pmod p\in \F_{p^\ell}[X]$ are linearly independent over $\F_{p^\ell}$. If $f_1,f_2,\ldots,f_s\in GR(p^a,\ell)[X]$ are linearly dependent over $GR(p^a,\ell)$, i.e. there exists $c_1,\ldots,c_s\in GR(p^a,\ell)$, not all zero such that $c_1f_1+c_2f_2+\ldots +c_sf_s=0$.  Let $i$ be the largest integer such that $p^i\mid c_1,\ldots,c_s$. Then we assume $c_j'=\frac{c_j}{p^i}$ for $j\in [s]$, which implies that $c_1',\ldots c_s'\pmod p$ are not all zero. Thus,
    \[
    c_1'f_1+c_2'f_2+\ldots +c_s'f_s=0\pmod p.
    \]
    This contradicts the assumption, and therefore $f_1,f_2,\ldots,f_s\in GR(p^a,\ell)[X]$ are linearly independent over $GR(p^a,\ell)$.
    
    On the other hand, we assume that $f_1,f_2,\ldots,f_s$ are linear independent over $GR(p^a,\ell)$. If $p \mid f_1,\ldots,f_s$, then $p^{a-1}f_1+p^{a-1}f_2+\ldots +p^{a-1}f_s=0$. Hence, $f_1,f_2,\ldots,f_s\pmod p$ are not all zero polynomial.
    
    If $f_1,f_2,\ldots,f_s\pmod p$ are linearly dependent over $\F_{p^\ell}$, i.e. there exists $c_1,c_2,\ldots,c_s\in \F_{p^\ell}$ such that:
    \[
    c_1f_1+c_2f_2+\ldots +c_sf_s=0\pmod p.
    \]
    Therefore, we obtain:
    \[
     c_1p^{a-1}f_1+c_2p^{a-1}f_2+\ldots +c_sp^{a-1}f_s=0
    \]
    This contradicts the assumption, and therefore $f_1,f_2,\ldots,f_s \pmod p\in \F_{p^\ell}[X]$ are linearly independent over $\F_{p^\ell}$.
    \label{lemlinear}
\end{proof}

\begin{lemma}\label{lm:det}
    Let $\F_{p^\ell}$ be the residue field of $GR(p^a,\ell)$ and $\gamma\in \F_{p^\ell}^{*}$ be a generator. The polynomials $f_1,f_2,\ldots,f_s\in GR(p^a,\ell)[X]^{<Rn}$ are linearly independent over $GR(p^a,\ell)$ if and only if determinant
    \[
\begin{pmatrix}
f_1(X) & f_2(X) & \cdots & f_s(X) \\
f_1(\gamma X) & f_2(\gamma X) & \cdots & f_s(\gamma X) \\
\vdots & \vdots & \ddots & \vdots \\
f_1(\gamma^{s-1} X) & f_2(\gamma^{s-1} X) & \cdots & f_s(\gamma^{s-1} X)
\end{pmatrix}
\]
is non-zero as a polynomial in $\F_{p^\ell}[X]$ i.e. it remains a nonzero polynomial when modulo $p$.
\end{lemma}

\begin{proof}
    According to \cref{lm:equiv}, $f_1,f_2,\ldots,f_s\in GR(p^a,\ell)[X]$ are linearly independent over $GR(p^a,\ell)$ and $f_1,f_2,\ldots,f_s\pmod p\in \F_{p^\ell}[X]$ are linearly independent over $\F_{p^\ell}$ are equivalent. Thus, the proof of the problem can be established by referring to the results in \cite{guruswami2016explicit}.
\end{proof}

\begin{thm}
Let $\cC_{FRS}$ be an $m$-Folded Reed-Solomon code of blocklength $N=\frac{n}{m}$ and rate $R$ over alphabet $GR(p^a,\ell)$. Suppose $\mathcal{H}$ be a free module of $GR(p^a,\ell)[X]^{<Rn}$ with rank $s$, and $r_i$ denotes the McCoy rank of the matrix $A_i$ associated with the $i$-th coordinate position, as defined above. Then we have:
\[
 \sum_{i=1}^{N}(s-r_i)\leq \frac{s\cdot Rn}{m-s+1}.
\]
\end{thm}

\begin{proof}
    From \cref{lm:det}, the determinant of:
\[
H(X) := 
\left[
\begin{array}{cccc}
h_1(X) & h_2(X) & \cdots & h_s(X) \\
h_1(\gamma X) & h_2(\gamma X) & \cdots & h_s(\gamma X) \\
\vdots & \vdots & \ddots & \vdots \\
h_1(\gamma^{s-1} X) & h_2(\gamma^{s-1} X) & \cdots & h_s(\gamma^{s-1} X)
\end{array}
\right] \pmod p
\]
is non-zero as $h_1,\ldots,h_s$ are linearly independent over $GR(p^a,\ell)$. Denote this determinant by $D(X)=\det(H(X))(mod\,p)$. Since each $h_i$ is of degree at most $Rn$, we note that $D(X)$ is a polynomial of degree at most $sRn$, by  \cref{lm:roots}, the number of zeros of $D(X)$ (with multiplicity) is bounded by $sRn$. Therefore, it suffices to show that the number of $D(X)$ is at least $(m-s+1)\cdot \sum_{i=1}^{N}(s-r_i)$.

In fact, we will describe the exact set of zeros with their multiplicities that illustrates this. The next claim immediately completes the proof. Note that we say that a non-root is a root with multiplicity $0$.

\begin{claim}
   For every $i\in [N]$, for every $j\in [m-s+1]$, $\gamma^{(i-1)m+j-1}$ is a root of $D(X)$ with multiplicity at least $s-r_i$.
\end{claim}

\begin{proof}
    Recall that $r_i$ is the McCoy rank of matrix $A_i$. For $j \in [m - s + 1]$, let $A_{ij}$ denote the $s \times s$ submatrix of $A_i$ formed by selecting all $s$ columns and rows from $j$ to $j + s - 1$. That is,
\[
A_{ij} = 
\begin{bmatrix}
h_1(\gamma^{(i-1)m+j-1}) & h_2(\gamma^{(i-1)m+j-1}) & \cdots & h_s(\gamma^{(i-1)m+j-1}) \\
h_1(\gamma^{(i-1)m+j}) & h_2(\gamma^{(i-1)m+j}) & \cdots & h_s(\gamma^{(i-1)m+j}) \\
\vdots & \vdots & \ddots & \vdots \\
h_1(\gamma^{(i-1)m+j+s-2}) & h_2(\gamma^{(i-1)m+j+s-2}) & \cdots & h_s(\gamma^{(i-1)m+j+s-2})
\end{bmatrix}
\]

Since $A_{ij}$ is a submatrix of $A_i$, $rank_M(A_{ij})\leq rank_M(A_i)=r_i$. If $r_i < s$, then $A_{ij}$ is not full rank and $p\mid \det(A_{ij})$. However, note that $A_{ij} = H(\gamma^{(i-1)m+j-1})$ and $\det(A_{ij})\pmod p=D(\gamma^{(i-1)m+j-1})=0$. Thus, if $s - r_i > 0$, then $\gamma^{(i-1)m+j-1}$ is a root of $D(X)$.

Extending this argument to multiplicities, let $D^{(k)}(X)$ be the $k$-th derivative of $D(X)$ for $k \in \{0, 1, \cdots, s\}$. Then this derivative can be written as a sum of $s^k$ determinants such that every determinant has at least $s - l$ columns common with $H(X)$. This follows by writing out the determinant as a signed sum of monomials, applying the product rule of differentiation, and packing them back into determinants.

Therefore, $D^{(k)}(\gamma^{(i-1)m+j-1})$ can be written as a sum of determinants where each determinant has at least $s - k$ columns in common with $A_{ij}$. For $k = 0,1,\dots,s - r_i - 1$, this leaves at least $r_i + 1$ columns in each determinant from $A_{ij}$. Recall that $rank_M(A_{ij}) \leq r_i$, which implies that the determinant of any $r+1$-th order submatrix in $A_{ij}$ is a zero-divisor or zero, causing each of the $s^k$ determinants in the sum for $H^{(k)}(\gamma^{(i-1)m+j-1})$ to vanish. We conclude that $H^{(k)}(\gamma^{(i-1)m+j-1}) = 0$ for $k = 0,1,\dots,s - r_i - 1$, and so $\gamma^{(i-1)m+j-1}$ is a root of $D(X)$ with multiplicity at least $s - r_i$.

Notice that $\gamma^{(i-1)m+j-1}$ is unit in $GR(p^a,\ell)$, we can obtain:
\[
\sum_{i=1}^{N}(s-r_i)\leq \frac{s\cdot Rn}{m-s+1}.
\]
\end{proof}
\end{proof}
We now show the theorem on the upper bound of list size using the induction method.

\begin{thm}
Let $\mathcal{C}_{FRS}$ be an $m$-folded Reed-Solomon code of blocklength $N = n/m$ and rate $R$. Suppose $s, b, m$ are integers such that $b > s$ and $m \geq b$. Then, for any $\vec{y} \in (GR(p^a,\ell)^m)^N$ and for every free module $\mathcal{H} \subseteq \mathcal{C}_{FRS}$ of rank $s$,
\[
\left| \mathcal{H} \cap \mathcal{L} \left( \vec{y}, \frac{b}{b+1} \cdot (1 - \frac{m}{m-b+1} \cdot R ) \right) \right| \leq (b-1) \cdot s + 1.
\]
\end{thm}

\begin{proof}
We prove this by induction on $s$. The case $s = 0$ is trivial, and the case $s = 1$ follows by \cref{s=1}. Using
\[
\left| \mathcal{H} \cap \mathcal{L} \left( \vec{y}, \frac{b}{b+1} \cdot ( 1 - \frac{m}{m-b+1} \cdot R ) \right) \right| \leq \left| \mathcal{H} \cap \mathcal{L} \left( \vec{y}, \frac{b}{b+1} \cdot (1 - R) \right) \right|.
\]
Henceforth, let $s \geq 2$, and denote $\mathcal{H}_{y} = \mathcal{H} \cap \mathcal{L} \left(\vec{y}, \frac{b}{b+1} \cdot \left( 1 - \frac{m}{m-b+1} \cdot R \right) \right)$, and $S_h$ be the agreement set between $\vec{y}$ and $\vec{h}$ (over all of $[N]$). Using the lower bound on the size of agreement sets,
\[
\left( \frac{1}{b+1} + \frac{bR}{b +1} \cdot \frac{m}{m-b+1} \right) N |\mathcal{H}_y| \leq \sum_{\vec{h} \in \mathcal{H}_y} |S_h|.
\]

An upper bound on $\sum_{\vec{h} \in \mathcal{H}_y} |S_h|$ can be proved using the inductive hypothesis. Again, we will consider two cases depending on $r_i = 0$ or $r_i > 0$. In the latter case, we can reduce dimension of the affine space $\mathcal{H}$ by $r_i > 0$ when we decide to assume $h_i = y_i$, so that the inductive hypothesis kicks in. Let $E \subseteq [N]$ be the bad set with $r_i = 0$, and $e = \frac{|E|}{N}$. It is easy to see that $e < R$.

For $i \in E$, we use the trivial bound $|\mathcal{H}_y|$ on the number of agreement sets $i$ belongs to. For $i \in \overline{E}$, the dimension reduces to $s - r_i$, and so the coordinate $i$ can appear in at most $(b - 1)(s - r_i) + 1$ many agreement sets.
\begin{align*}
\sum_{h \in \mathcal{H}_y} |S_h| &= \sum_{i=1}^{N} \left| \{ h \in \mathcal{H}_y : \forall j \in [m], \ h(\gamma^{(i-1)m+j-1}) = y(\gamma^{(i-1)m+j-1}) \} \right| \\
&\leq \sum_{i \in \overline{E}} [(b - 1)(s - r_i) + 1] + \sum_{i \in E} |\mathcal{H}_y| \\
&\leq |E| \cdot |\mathcal{H}_y| + N - |E| + (b - 1) \left( \frac{s \cdot Rn}{m - s + 1} - s|E| \right) \\
&\leq |E| \cdot |\mathcal{H}_y| + N \left( 1 - e + (b - 1)s ( \frac{m}{m - s + 1} R - e ) \right).
\end{align*}

Comparing the lower bound and upper bound,

\begin{align*}
|\mathcal{H}_y| &\leq \frac{1 - e + (b - 1)s \left( \frac{m}{m - s + 1} R - e \right)}{\left( \frac{1}{b+1} + \frac{bR}{b+1} \cdot \frac{m}{m-b+1} - e \right)} \\
&< \frac{1 - e + (b - 1)s \left( \frac{m}{m - b + 1} R - e \right)}{\left( \frac{1}{b+1} + \frac{bR}{b+1} \cdot \frac{m}{m-b+1} - e \right)}.
\end{align*}

We show that $|\mathcal{H}_y| < 1 + (b - 1)s$ by showing that

\[
\left( \frac{1}{b+1} + \frac{bR}{b+1} \cdot \frac{m}{m-b+1} - e \right) \left(|\mathcal{H}_y| - 1 - (b - 1)s\right) < 0.
\]
This suffices to conclude our induction.
\begin{align*}
&\left( \frac{1}{b+1} + \frac{bR}{b+1} \cdot \frac{m}{m-b+1} - e \right) \left(|\mathcal{H}_y| - 1 - (b - 1)s\right) \\
&< 1 + \frac{m}{m-b+1} (b-1)sR - \frac{1}{b+1} - \frac{bR}{b+1} \cdot \frac{m}{m-b+1} - \frac{(b-1)s}{b+1} - \frac{bR}{b+1} \cdot \frac{m}{m-b+1} \cdot (b-1)s \\
&= \left( \frac{b-(b-1)s}{b+1} \right) \cdot \left( 1 - \frac{m}{m-b+1} R \right).
\end{align*}

The last term is $\leq 0$ as long as $b \leq (b-1)s$, which is always true for $s \geq 2$.
\end{proof}

\begin{col}
    Let $\mathcal{C}_{FRS}$ be an $m$-folded Reed-Solomon code of blocklength $N=\frac{n}{m}$ and rate $R$. Nocite that in the \cref{freemodule}, we prove that the decoding list is confined within a free module. Hence, for any integer $s$, $1\leq s\leq m$, and for any $\vec{y}\in \left(GR(p^a,\ell)^m\right)^N$, it holds that:
\[
\left| \mathcal{L} \left( \vec{y}, \frac{b}{b+1} \left( 1 - \frac{m}{m-b+1} R \right) \right) \right| \leq (b-1)^2 + 1.
\]
\end{col}

\section{Optimal Combinatorial List Size for Folded Reed--Solomon Codes}\label{sec:optimal-combinatorial-list-size}
In Section~\ref{section:6}, we obtained an algorithmic list-size bound $O(1/\varepsilon^2)$ for folded Reed--Solomon codes over Galois rings by extending the module-theoretic approach of~\cite{srivastava2025improved}. In this section, we further show that the relaxed generalized Singleton bound with optimal list size $O(1/\varepsilon)$ also extends to the Galois-ring setting at the combinatorial level.

For $f\in GR(p^a,\ell)[X]_{<k}$, define its $m$-folded codeword by
\[
\mathrm{Enc}_{\mathrm{FRS}}(f)=\bigl(f[0],\ldots,f[N-1]\bigr)\in (GR(p^a,\ell)^{m})^{N},
\]
where $f[i]:=\bigl(f(\gamma^{im}),f(\gamma^{im+1}),\ldots,f(\gamma^{im+m-1})\bigr)\in GR(p^a,\ell)^{m}$ for $0\leq i\leq N-1$.

For an integer $L$ with $1\le L\le m$, define $\rho_{L}:=\frac{L}{L+1}\left(1-\frac{mR}{m-L+1}\right)$. Our goal is to prove the following theorem.

\begin{thm}\label{thm:optimal_combinatorial_galois}
Let $\cC_{\mathrm{FRS}}\subseteq (GR(p^a,\ell)^{m})^{N}$ be the $m$-folded Reed--Solomon code of block length $N=n/m$, and rate $R=k/n$. Then for every integer $L$ with $1\le L\le m$, the code $\cC_{\mathrm{FRS}}$ is $\left(\rho_{L},L\right)$ average-radius list-decodable.
\end{thm}

We now prove this theorem through a sequence of lemmas. Assume for contradiction that $\cC_{\mathrm{FRS}}$ is not $(\rho_{L},L)$ average-radius list-decodable. Then there exist pairwise distinct polynomials $f_{0},f_{1},\ldots,f_{L}\in GR(p^a,\ell)[X]_{<k}$ and a received word $\mathbf{y}=\bigl(\mathbf{y}[0],\mathbf{y}[1],\ldots,\mathbf{y}[N-1]\bigr)\in (GR(p^a,\ell)^{m})^{N}$ such that
\[
\frac{1}{L+1}\sum_{t=0}^{L} d\bigl(\mathrm{Enc}_{\mathrm{FRS}}(f_{t}),\mathbf{y}\bigr)
\le \rho_{L}\cdot N.
\]
For each folded coordinate $u\in [N]$, define $e_{i}:=\bigl\{t\in\{0,1,\ldots,L\}:f_t[i]=\mathbf{y}[i]\bigr\}$ and define the global weight $W:=\sum_{i=0}^{N-1}\max\{|e_{i}|-1,0\}$.

\begin{lemma}\label{lem:global_weight_lower_bound}
Under the above assumption, we have $W\ge \frac{Lk}{m-L+1}$.
\end{lemma}

\begin{proof}
Since $\sum_{t=0}^{L}\bigl(N-d(\mathrm{Enc}_{\mathrm{FRS}}(f_{t}),\mathbf{y})\bigr)
\ge (L+1)(1-\rho_{L})N$, we obtain
\[
\begin{aligned}
W &=\sum_{i=0}^{N-1}\max\{|e_{i}|-1,0\}
\ge \sum_{i=0}^{N-1}(|e_{i}|-1) \\
&=\sum_{t=0}^{L}\bigl(N-d(\mathrm{Enc}_{\mathrm{FRS}}(f_{t}),\mathbf{y})\bigr)-N\geq (L+1)(1-\rho_L)N-N
\end{aligned}
\]
Substituting $\rho_{L}=\frac{L}{L+1}\left(1-\frac{mR}{m-L+1}\right)$ and $R=\frac{k}{n}=\frac{k}{mN}$, we conclude that $W\ge \frac{Lk}{m-L+1}$.
\end{proof}

Before turning to our argument in detail, we first present a lemma from the finite-field setting in~\cite{chen2025explicit}, which will be used later to derive the desired contradiction.

\begin{lemma}[Lemma~2.15~\cite{chen2025explicit}]\label{lem:field-distinct-weight}
    For a received word $\mathbf{y}\in (\F_q^m)^N$, consider any $s\geq 2$ polynomials $f_1,\ldots,f_s\in \F_q[X]_{<k}$ and $e_i:=\{t\in \{1,\ldots,s\}:f_t[i]=\mathbf{y}[i]\}$ for $0\leq i\leq N-1$. If $\sum_{i=0}^{N-1}\max\{|e_i|-1,0\}\geq \frac{(s-1)k}{m-s+2}$, then polynomials $f_1,\ldots,f_s$ cannot be distinct.
\end{lemma}

\begin{col}\label{col:field-distinct-weight}
    Under the assumption in Lemma~\ref{lem:field-distinct-weight} and let $T := \{ i \in \{0,1,\ldots,N-1\} : e_i \neq \emptyset \}$. If $\sum_{i\in T}\max\{|e_i|-1,0\}\geq \frac{(s-1)k}{m-s+2}$, then polynomials $f_1,\ldots,f_s$ cannot be distinct.
\end{col}
\begin{proof}
For every $i\notin T$, we have $e_i=\emptyset$, and hence $\max\{|e_i|-1,0\}=0$. Therefore $\sum_{i=0}^{N-1}\max\{|e_i|-1,0\}=\sum_{i\in T}\max\{|e_i|-1,0\}$. Thus the assumption of the corollary is exactly the assumption of Lemma~\ref{lem:field-distinct-weight} and applying Lemma~\ref{lem:field-distinct-weight} immediately shows that polynomials $f_1,\ldots,f_s$ cannot be distinct.
\end{proof}

To make use of the finite-field lemma in the Galois-ring setting, we organize the
bad list according to its successive $p$-adic congruence classes. The tree produced by Algorithm~\ref{alg:padic-tree} will be referred to as the $p$-adic congruence tree associated with the bad list$\{f_0,\ldots,f_L\}$. Its level-$r$ vertices are precisely the classes in $\Pi_r$, and for each class $C \in \Pi_r$, the children of $C$ are exactly the classes in $\Pi_{r+1}$ contained in $C$. Moreover, since $f_0,\ldots,f_L$ are pairwise distinct, $\Pi_a$ is discrete, and thus the leaves of $T$ are precisely the singleton classes. Unary nodes may occur, and we keep them explicitly.

\begin{algorithm}[t]
\caption{Construction of the $p$-adic congruence tree}
\label{alg:padic-tree}
\begin{algorithmic}[1]
\Statex \textbf{Input:} Pairwise distinct polynomials
$f_0,f_1,\ldots,f_L \in GR(p^a,\ell)[X]_{<k}$.
\Statex \textbf{Output:} A rooted tree $T$.

\State Let the partition
\[
\Pi_0:=\big\{\{0,1,\ldots,L\}\big\},
\]
and let the set $\{0,1\ldots,L\}$ be the root of $T$.

\For{$r=0,1,\ldots,a-1$}
    \State Construct the partition $\Pi_{r+1}$ from $\Pi_r$ as follows.
    \For{each set $C\in \Pi_r$}
        \State Partition $C$ into a disjoint union of subsets
        \[
        C=C_1\cup C_2\cup \cdots \cup C_h
        \]
        such that, for any $s,t\in C$,
        \[
        s \text{ and } t \text{ lie in the same set } C_j
        \quad\Longleftrightarrow\quad
        f_s \equiv f_t \pmod{p^{r+1}}.
        \]
        \State Since $f_s, s\in C_j$ has the same prefix $h:=f_s \pmod{p^{r+1}}$, let $t_{C_j}\in C$ be the representative element of set $C$ which implies $f_{t_{C_j}}=h \pmod{p^{r+1}}$.
        \State Add the sets $C_1,\ldots,C_h$ to $\Pi_{r+1}$.
        \State Add each $C_j$ as a child of $C$ in $T$ and set $\mathrm{ch}(C)=\{C_1,\ldots,C_h\}$.
    \EndFor
\EndFor

\State \Return $T$.
\end{algorithmic}
\end{algorithm}

Let $\mathrm{Int}(T)$ denote the set of all internal nodes of $T$. For each $C\in \mathrm{Int}(T)$ on level $r$, we fix a representative element $t_C\in C$. If $\mathrm{ch}(C)=\{D_1,D_2,\ldots,D_c\}$, then for each child $D\in \mathrm{ch}(C)$ we also fix a representative element $t_D\in D$. These choices allow us to define the local residual polynomials below.

\begin{defn}
For each child $D\in \mathrm{ch}(C)$, define the local residual polynomial
\[
g_{D}(X):=\frac{f_{t_{D}}(X)-f_{t_{C}}(X)}{p^{r}}\pmod p\in \F_{p^{\ell}}[X]_{<k}.
\]
\end{defn}

\begin{lemma}\label{lem:fieldization_node}
Let $C\in \mathrm{Int}(T)$ be an internal node on level $r$. Then
\begin{enumerate}
    \item The polynomial $g_{D}$ is well defined for every child $D\in \mathrm{ch}(C)$.
    \item The family $\{g_{D}:D\in \mathrm{ch}(C)\}$ consists of pairwise distinct polynomials in $\F_{p^{\ell}}[X]_{<k}$.
    \item For each folded coordinate $i\in \{0,\ldots,N-1\}$, define $e_i(C):=\{D\in \mathrm{ch}(C): D\cap e_{i}\neq \emptyset \}$. If $e_i(C)\neq \emptyset$, then 
    \[
    \mathbf{y}'[i]:=\frac{\mathbf{y}[i]-f_{t_{C}}[i]}{p^{r}}\pmod p\in \F_{p^{\ell}}^{m}
    \]
    is well defined. Moreover, for every $D\in e_i(C)$ we have $g_D[i]=\mathbf{y}'[i]$ where $g_D[i]$ is the evaluation of polynomial $g_D(X)$ at the $i$-th folded block over the residue field $\F_{p^{\ell}}$.
\end{enumerate}
\end{lemma}

\begin{proof}
Since $t_{D},t_{C}\in C\in \Pi_{r}$, we have $f_{t_{D}}-f_{t_{C}}\in p^{r}\cdot GR(p^a,\ell)[X]$, so $g_{D}$ is well defined and still has degree $<k$.

If $g_{D}=g_{D'}$ for two distinct children $D\neq D'$, then
\[
\frac{f_{t_{D}}-f_{t_{C}}}{p^{r}}
\equiv
\frac{f_{t_{D'}}-f_{t_{C}}}{p^{r}}
\pmod p,
\]
This implies $f_{t_{D}}-f_{t_{D'}}\in p^{r+1}\cdot GR(p^a,\ell)[X]$ which contradicts that $D$ and $D'$ are distinct children of $C$. Thus, $g_{D}$ are pairwise distinct.

Assume now that $e_i(C)\neq \emptyset$, and choose $D\in e_i(C)$ together with $h\in D\cap e_i$. Then $f_h[i]=\mathbf{y}[i]$. Since $h,t_{C}\in C$, we have $f_{h}-f_{t_{C}}\in p^{r}\cdot GR(p^a,\ell)[X]$, and therefore $\mathbf{y}[i]-f_{t_C}[i]=(f_h-f_{t_C})[i]\in p^{r}\cdot GR(p^a,\ell)^{m}$. Hence, $\mathbf{y}'[i]$ is well defined.

Fix any $D\in e_i(C)$ and choose $h\in D\cap e_i$. Since $h,t_{D}\in D\in \Pi_{r+1}$, we have $f_{h}-f_{t_{D}}\in p^{r+1}\cdot GR(p^a,\ell)[X]$. Thus, $\frac{f_{h}-f_{t_{C}}}{p^{r}}\equiv\frac{f_{t_{D}}-f_{t_{C}}}{p^{r}}\pmod p$, which implies
\[
g_{D}[i]\equiv\frac{f_h[i]-f_{t_C}[i]}{p^{r}}\equiv\frac{\mathbf{y}[i]-f_{t_C}[i])}{p^{r}}\equiv\mathbf{y}'[i] \pmod p.
\]
This proves the lemma.
\end{proof}

For a node $C$ and a folded coordinate $i$, define $\widetilde{e}_{i}(C):=\{D\in \mathrm{ch}(C): g_D[i]=\mathbf{y}'[i]\}$. Lemma~\ref{lem:fieldization_node} shows that $e_i(C)\subseteq \widetilde{e}_{i}(C)$. Define the corresponding weight by $W(C):=\sum_{i=0}^{N-1}\max\{|e_i(C)|-1,0\}$ and $\widetilde{W}(C):=\sum_{u:\,e_i(C)\neq \emptyset}\max\{|\widetilde{e}_{i}(C)|-1,0\}$. Then, we have $W(C)\le \widetilde{W}(C)$.

\begin{lemma}\label{lem:local_upper_bound}
Let $C\in \mathrm{Int}(T)$ be an internal node. If $|\mathrm{ch}(C)|=1$, then $W(C)=0$. If $|\mathrm{ch}(C)|\ge 2$, then $W(C)<\frac{(|\mathrm{ch}(C)|-1)k}{m-|\mathrm{ch}(C)|+2}$.
\end{lemma}

\begin{proof}
If $|\mathrm{ch}(C)|=1$, then for every $i$ we have $|e_i(C)|\in\{0,1\}$, so $W(C)=0$. Assume now that $|\mathrm{ch}(C)|\ge 2$. For each coordinate $i$ with $e_i(C)\neq \emptyset$, the vector $\mathbf{y}'[i]\in \F_{p^{\ell}}^{m}$ is well defined, and every child in $e_i(C)$ agrees with $\mathbf{y}'[i]$ on the $i$-th folded block. Thus, we can apply Corollary~\ref{col:field-distinct-weight} over $\F_{p^\ell}$ to polynomial set  $\{g_{D}: D\in \mathrm{ch}(C)\}$  and the vector
$(\mathbf{y}'[i])_{i:e_i(C)\neq \emptyset}$. By Lemma~\ref{lem:fieldization_node}, the polynomials $\{g_{D}:D\in \mathrm{ch}(C)\}\subseteq \F_{p^{\ell}}[X]_{<k}$ are pairwise distinct. By Corollary~\ref{col:field-distinct-weight}, we conclude that $\widetilde{W}(C)<\frac{(|\mathrm{ch}|-1)k}{m-|\mathrm{ch}|+2}$. Therefore,
\[
W(C)\le \widetilde{W}(C)<\frac{(|\mathrm{ch}(C)|-1)k}{m-|\mathrm{ch}(C)|+2}.
\]
This completes the proof.
\end{proof}

\begin{rmk}\label{rem:zero_vertex}
If the child containing the base representative $t_{C}$ is included in $\mathrm{ch}(C)$, then the corresponding residual polynomial is $g_{D}=0$. This causes no difficulty: the finite-field lemma is used exactly in the same 'translate one polynomial to zero polynomial' manner as in the proof of~\cite{chen2025explicit}.
\end{rmk}

\begin{lemma}\label{lem:weight_decomposition}
Under the above assumption, we have
\[
W=\sum_{C\in \mathrm{Int}(T)} W(C).
\]
\end{lemma}

\begin{proof}
Fix a folded coordinate $i\in [N]$. If $e_i=\emptyset$, then the contribution of $i$ to both sides is zero. Assume now that $e_i\neq \emptyset$, consider the minimal rooted subtree $T_{u}$ of $T$ spanned by the leaves indexed by $e_i$.

For each internal node $C$ of $T$, the number $|e_i(C)|$ is exactly the number of children of $C$ that belong to $T_{u}$. Therefore, $\max\{|e_i(C)|-1,0\}$ records the contribution of $C$ to the leaf-count identity on $T_{u}$. We claim that
\begin{equation}\label{eq:coordinate_decomposition}
\max\{|e_i|-1,0\}=\sum_{C\in \mathrm{Int}(T)} \max\{|e_i(C)|-1,0\}.
\end{equation}
Indeed, if $T_{u}$ has only one selected leaf, then both sides of~\eqref{eq:coordinate_decomposition} are zero. Otherwise, let the root of $T_{u}$ have $d$ selected child subtrees with leaf counts $s_{1},\ldots,s_{d}\ge 1$. By induction on the height of $T_{u}$, the contribution of the $j$-th selected child subtree is $s_{j}-1$, while the root contributes $d-1$. Hence, $(d-1)+\sum_{j=1}^{d}(s_{j}-1)=\sum_{j=1}^{d}s_{j}-1=|e_i|-1$ which proves~\eqref{eq:coordinate_decomposition}. Summing~\eqref{eq:coordinate_decomposition} over all $i\in \{0,1,\ldots,N-1\}$ gives
\[
W=\sum_{i=0}^{N-1}\max\{|e_i|-1,0\}=\sum_{i=0}^{N-1}\sum_{C\in \mathrm{Int}(T)}\max\{|e_i(C)|-1,0\}=\sum_{C\in \mathrm{Int}(T)}W(C).
\]
This completes this proof.
\end{proof}

\begin{lemma}\label{lem:tree_threshold}
For any rooted tree $T$ with $L+1$ leaves,
\[
\sum_{C\in \mathrm{Int}(T)} \frac{|\mathrm{ch}(C)|-1}{m-|\mathrm{ch}(C)+2|}\leq \frac{L}{m-L+1}.
\]
\end{lemma}

\begin{proof}
Let $\psi(x):=\frac{x}{m+1-x}$ where $0\leq x\leq m$. We first show that $\psi$ is superadditive on its domain: if $x,y\ge 0$ and $x+y\le m$, then
\begin{align*}
\psi(x+y)-\psi(x)-\psi(y)
&=\frac{x+y}{m+1-x-y}-\frac{x}{m+1-x}-\frac{y}{m+1-y}\\
&=\frac{xy(2m+2-x-y)}{(m+1-x-y)(m+1-x)(m+1-y)}\ge 0.
\end{align*}
Hence $\psi(x+y)\ge \psi(x)+\psi(y)$, whenever $x,y\ge 0$ and $x+y\le m$.

We now prove the lemma by induction on the number of leaves. The case of a single leaf is trivial. Suppose the root has $c$ children and the corresponding child subtrees have leaf counts $n_{1},\ldots,n_{c}$, so that $\sum_{i=1}^{c} n_{i}=L+1$.

By the induction hypothesis, the total contribution of the $i$-th child subtree is at most $\phi(n_{i})$. Therefore,
\[
\begin{aligned}
\sum_{C\in \mathrm{Int}(T)}\frac{|\mathrm{ch}(C)|-1}{m-|\mathrm{ch}(C)|+2}
&\leq \frac{c-1}{m-c+2}+\sum_{i=1}^{c}\frac{n_i-1}{m-n_i+2}=\psi(c-1)+\sum_{i=1}^c \psi(n_i-1) \\
&\leq \psi\bigl((c-1)+\sum_{i=1}^c n_i-1\bigr)\leq \psi(L)=\frac{L}{m-L+1}.
\end{aligned}
\]
This proves the lemma.
\end{proof}

Now we can prove the main Theorem~\ref{thm:optimal_combinatorial_galois}.

\begin{proof}[Proof of Theorem~\ref{thm:optimal_combinatorial_galois}]
Assume that there exists a bad list. By Lemma~\ref{lem:global_weight_lower_bound}, $W\ge \frac{Lk}{m-L+1}$. On the other hand, Lemma~\ref{lem:weight_decomposition} gives $W=\sum_{C\in \mathrm{Int}(T)}W(C)$. Since the tree $T$ has $L+1\ge 2$ leaves, there exists at least one effective internal node with at least two children. Applying Lemma~\ref{lem:local_upper_bound} and Lemma~\ref{lem:tree_threshold}, we have $W<\frac{Lk}{m-L+1}$, which is contradiction. Therefore, no such bad list exists and $\cC_{\mathrm{FRS}}$ is $(\rho_{L},L)$ average-radius list-decodable.
\end{proof}

\begin{rmk}
As an immediate consequence, by choosing
\[
L=\Theta(1/\varepsilon),\qquad m=\Theta(1/\varepsilon^{2}),
\]
we obtain the combinatorial conclusion
\[
(1-R-\varepsilon, O(1/\varepsilon))
\]
for average-radius list-decodability of folded Reed--Solomon codes over Galois rings.
\end{rmk}

\bibliographystyle{alpha}
\bibliography{refs}
\section{Appendix}
To keep the paper self-contained, we provide in the appendix a proof that the Johnson bound remains valid over Galois rings.

\subsection{Johnson bounds for Galois Ring}
\begin{lemma}[Zarankiewicz Theorem]\label{lem:zarankiewicz}
    Let $G=(L,R,E)$ be a bipartite graph with $|L|=l$ and $|R|=r\geq 2$. For any $s\leq l$, we say $G$ is $K_{s,2}$ free if there is no subset $L'\subseteq L$ and $R'\subseteq R$ with $|L'|=s$ and $|R'|=2$ such that $L'\times R'\subseteq E$. If $G$ is $K_{s,2}$ free then 
    \[
        |E|\leq l+r\sqrt{(s-1)l}
    \]
\end{lemma}

\begin{proof}

Define an $l\times r$ matrix $M$ that is the adjacency matrix of $G$ i.e. each row and column of $M$ is indexed by a vertex in $L$ and $R$ respectively and for any $(u,w)\in L\times R, M_{u,w}=1$ iff $(u,w)\in E$. Define $v=\sum_{w\in R}M^w$, where recall the $M^w$ is the $w$-th of $M$. 

Consider the similarity of the edges between two fixed vertices in $R$, that is, how many vertices in $L$ are simultaneously connected to these two vertices in $R$. Let the sum of the similarities of the edges between any two vertices in $R$ be $S$. Assuming there are $m_i$ non-zero elements in the $i$-th row of the matrix $M$, then we have $S=\frac{\sum_{i=1}^{l}m_i(m_i-1)}{2}$. As $G$ is $K_{s,2}$ free, thus, $S\leq \frac{r(r-1)(s-1)}{2}$. By the Cauchy-Schwarz Inequality, we can obtain:
\[
\sum_{i=1}^{l}m_i^2\cdot l\geq (\sum_{i=1}^{l}m_i)^2 =|E|^2
\]
Therefore, $ |E|\leq l+r\sqrt{(s-1)l}$.
\end{proof}

\begin{thm}[Alphabet-Free Johnson Bound]
    For every code $\cC$ with block length $n$ and distance $d$ over $GR(p^a,\ell)$, if $e< n-\sqrt{n(n-d)}$, then the code is $(\frac{e}{n},n)$-list decodable.
\end{thm}

\begin{proof}
    Let $\cC \subseteq GR(p^a,\ell)^n $ be a code of distance $d$, $y\in GR(p^a,\ell)^n$ and $c_1,c_2,...,c_L$ be distinct codewords in $\cC$ such that $d(y,c_i)\leq n-\sqrt{n(n-d)}-1$ for every $i\in [L]$. Define a graph $G=([n],[L],E)$ to be a bipartite graph such that $(i,j)\in [n]\times [L]$ is an edge iff $y_i=(c_j)_i$.
    As $d(y,c_i)\leq n-\sqrt{n(n-d)}-1$ for every $i\in [L]$, 
    \[
    |E|\geq L(\sqrt{n(n-d)}+1)
    \]
    Any $c_i\neq c_j\in \cC$ cannot be the same as the vector $y$ in $n-d+1$ positions, otherwise $d(c_i,c_j)<d$. Hence, the graph $G$ is $K_{n-d+1,2}$ free. By Lemma~\ref{lem:zarankiewicz}, 
    \[
    L(\sqrt{n(n-d)}+1)\leq n+L(\sqrt{n(n-d)}),\qquad  L\leq n.
    \]
\end{proof}

\end{document}